\documentclass{article}
\usepackage{graphicx} %
\usepackage{enumitem}
\usepackage[T1]{fontenc}
\usepackage[linesnumbered,ruled,vlined]{algorithm2e}
\usepackage{notation}

\usetikzlibrary{arrows.meta,positioning,calc}

\SetKwInput{KwInput}{Input}
\SetKwInput{KwOutput}{Output}
\SetKwInput{KwReturn}{Return}

\usepackage{arxiv}

\title{Computing Equilibrium beyond Unilateral Deviation}
\author{Mingyang Liu $^1$, Gabriele Farina$^1$, Asuman Ozdaglar $^1$\\
    $^1$ LIDS, EECS, Massachusetts Institute of Technology\\
    $^1$ \texttt{\{liumy19,gfarina,asuman\}@mit.edu}\\
}

\begin{document}

\maketitle

\begin{abstract}

Most familiar equilibrium concepts, such as Nash and correlated equilibrium, guarantee only that no single player can improve their utility by deviating unilaterally. They offer no guarantees against profitable coordinated deviations by coalitions. Although the literature proposes solution concepts that provide stability against multilateral deviations (\emph{e.g.}, strong Nash and coalition-proof equilibrium), these generally fail to exist. In this paper, we study an alternative solution concept that minimizes coalitional deviation incentives, rather than requiring them to vanish, and is therefore guaranteed to exist. Specifically, we focus on minimizing the average gain of a deviating coalition, and extend the framework to weighted-average and maximum-within-coalition gains. In contrast, the minimum-gain analogue is shown to be computationally intractable. For the average-gain and maximum-gain objectives, we prove a lower bound on the complexity of computing such an equilibrium and present an algorithm that matches this bound. Finally, we use our framework to solve the \emph{Exploitability Welfare Frontier} (EWF), the maximum attainable social welfare subject to a given exploitability (the maximum gain over all unilateral deviations).

\end{abstract}

\section{Introduction}
\label{section:intro}

Most equilibrium concepts studied so far, such as Nash equilibrium (NE) \citep{nash1950equilibrium-nash-def}, correlated equilibrium (CE) \citep{aumann1974subjectivity-CE-def}, coarse correlated equilibrium (CCE) \citep{moulin1978strategically-CCE-def}, and Stackelberg equilibrium \citep{von2010market-stackelberg-def}, guarantee only that \emph{no individual player} has a unilateral utility improving deviation. However, they offer no guarantees when multiple players deviate simultaneously by forming a coalition. In this paper, we address the following question:
\begin{quote}\itshape 
What is an appropriate notion for capturing multilateral deviations, and how can it be computed efficiently?
\end{quote}

Previous notions that address coalition deviations, such as strong NE \citep{strong-NE}\footnote{In this paper, we use the term \emph{strong equilibrium} to broadly refer to any equilibrium concept that considers multilateral deviations.} and coalition-proof equilibrium \citep{bernheim1987coalition-proof}, mostly fail to exist in general games (unlike NE). Therefore, instead of searching for a joint strategy immune to all coalition deviations, we focus on computing a joint strategy that \emph{minimizes} the maximum average gain\footnote{Our framework also extends to coalition objectives based on weighted sums or weighted averages of players' deviation gains, as well as objectives based on the maximum within-coalition deviation gain. In contrast, the corresponding minimum-gain variant is computationally intractable even for small games. See \Cref{theorem:strong-CCE-hardness,remark:extension} for details.} achievable by any coalition through joint deviation. In other words, we aim to compute the most stable strategy profile, even if a perfectly stable one does not exist. We refer to this notion as the \emph{Minimum Average-Strong Equilibrium} (MASE).

The difficulty of this optimization problem naturally depends on the complexity of the interactions between players. To formalize this, we introduce the \Cref{def:utility-dependence-graph}, $\cG(\cV,\cE)$, where each player is a vertex. An edge connects two players, $i$ and $j$, if and only if there is some player $k$ whose utility is affected by the actions of both $i$ and $j$. Intuitively, an edge $\rbr{i,j}$ indicates that the actions of $i$ and $j$ jointly influence some player's utility, and hence that coalitional deviations involving these players can create nontrivial joint effects. The graph $\cG$ provides a clear map of the game's interaction structure, and its properties can help us understand the computational complexity of finding the MASE. For games with simple interaction structures (\emph{e.g.}, a sparse \Cref{def:utility-dependence-graph}), one might expect to compute the MASE efficiently.

We first show that computing the MASE is computationally challenging in the general case. We establish two key hardness results that delineate the problem's complexity.

First, the problem is fundamentally harder than finding equilibria like NE or CE. In those cases, we are solving a feasibility problem: finding a strategy where the maximum gain from deviating is at most zero. For MASE, we must solve an optimization problem: minimizing this maximum gain. This distinction is crucial, and we show that even for the simplest case of single-player deviations (\emph{i.e.}, coalitions of size one)\footnote{Restricting to singleton coalitions aligns the deviation model with unilateral deviations, but it still differs from NE: NE only asks whether there exists a feasible solution with deviation gain $\leq 0$, whereas MASE minimizes the deviation gain}, approximating the MASE value to within a factor that is inverse polynomial in the number of players is {\sf NP}-hard. This indicates that even without considering complex coalitions, the problem is intractable without additional assumptions on the game's structure.

Second, we show that this complexity is intrinsically tied to the structure of the \Cref{def:utility-dependence-graph}. Building on the strong exponential time hypothesis (SETH) \citep{impagliazzo2001complexity-SETH} (see \Cref{theorem:tree-width-important} for details), we prove that solving MASE requires time that is at least exponential in the treewidth\footnote{Treewidth can be thought of as a formal measure of how sparse and "tree-like" a graph is.} of the \Cref{def:utility-dependence-graph}. This holds even when we only consider coalitions of a constant size. This result demonstrates that the treewidth is a fundamental barrier, and an exponential dependence on it is unavoidable.

Finally, we complement the lower bound with an algorithm whose running time has exponential dependence only on the treewidth of the \Cref{def:utility-dependence-graph}, rather than on the number of players. This demonstrates that our hardness result is tight and establishes the treewidth as the definitive parameter characterizing the complexity of computing the MASE. While the problem is hard in general, it becomes tractable for games where the underlying interaction structure is not too complex. The key idea of our algorithm involves breaking the game into small, overlapping, weakly interacting parts, solving the problem locally on each part, and then combining these local solutions consistently.

To summarize, our contributions are as follows:
\begin{enumerate}[nosep]
    \item \textbf{Complexity characterization.} We establish lower bounds on the computational complexity of computing MASE, showing that an exponential dependence on the treewidth of the \Cref{def:utility-dependence-graph} is unavoidable (\Cref{theorem:tree-width-important}). We further show that the minimum-gain variant over coalitions is computationally intractable even for small games (\Cref{theorem:strong-CCE-hardness}).
    
    \item \textbf{Algorithmic contribution.} We design an algorithm for computing MASE whose running time matches the lower bound in its exponential dependence on the treewidth of the \Cref{def:utility-dependence-graph}, rather than depending exponentially on the number of players.

    \item \textbf{Application to welfare-exploitability tradeoffs.} We show how our framework can be used to compute socially optimal strategies subject to a prescribed exploitability budget, thereby characterizing the exploitability welfare frontier (\Cref{section:tradeoff}).
\end{enumerate}

\paragraph{Technical Overview.}
In \Cref{section:computation-MASE}, we formulate the strong-equilibrium computation problem as a two-player zero-sum meta-game between a correlator and a deviator. The correlator chooses a correlated joint strategy for all players in the original game, whereas the deviator chooses a coalition together with a joint deviation for that coalition. The game is zero-sum because the correlator seeks to minimize the coalition's gain from deviation, while the deviator seeks to maximize it.

The main computational challenge is that the action spaces of both players in this meta-game are exponentially large: the correlator's actions are joint action profiles, and the deviator's actions are coalition-deviation pairs. Standard game-solving approaches, such as no-regret learning \citep{hazan2016introduction} and linear programming \citep{papadimitriou2008computing-against-hope}, typically scale with the largest action space among the players, so applying them directly is computationally intractable. To overcome this obstacle, we instantiate no-regret learning with Follow the Perturbed Leader (FTPL) \citep{hazan2016introduction}. Each FTPL update reduces to optimizing a linear objective over the relevant simplex, whose optimum is attained at an extreme point. Thus, it suffices to compute and store a pure strategy at each iteration rather than maintain the full exponential-dimensional distribution.

Finally, we solve the resulting linear optimization oracle via dynamic programming over a tree decomposition of the \Cref{def:utility-dependence-graph}. This yields a running time that is exponential only in the treewidth of the \Cref{def:utility-dependence-graph}, rather than in the number of players. Consequently, the algorithm is efficient for games whose \Cref{def:utility-dependence-graph} has bounded treewidth. Moreover,  its exponential dependence on the treewidth matches the lower bound in \Cref{theorem:tree-width-important}.

\section{Related Work}
\label{section:related-work}

In this section, we review the literature on strong equilibrium from three perspectives: existence, time complexity, and computation methods.

\paragraph{Existence of Strong Equilibrium.} \citet{strong-NE} introduced the strong NE, where no coalition (a nonempty subset of players) can deviate in a way that strictly improves the utility of all its members. However, even in simple two-player games such as Prisoner’s Dilemma \citep{prisoner-dilemma}, a strong NE does not exist when players can deviate simultaneously. To address this, \citet{bernheim1987coalition-proof} proposed the coalition-proof equilibrium, which restricts the set of deviations (see \Cref{section:relation} for details). Yet, this concept also fails to guarantee existence, already in three-player games \citep{bernheim1987coalition-proof}. More recently, \citet{rahn2015efficient-collaborative-game} studied the notion of $\alpha$-approximate $k$-equilibrium, where no coalition of size at most $k$ can deviate so that every member’s utility becomes at least $\alpha\geq 1$ times their original utility. They further showed that such equilibria exist in graph coordination games only under specific conditions, for instance, when $\alpha \geq 2$. Motivated by these non-existence results, we instead focus on minimizing the maximum average gain from coalition deviations \Cref{eq:MASE}, a quantity that is always well defined.

\paragraph{Complexity of Strong Equilibrium.} Since a strong NE degenerates to an NE when only singleton coalitions are considered, computing a strong NE is {\sf PPAD}-hard in general \citep{daskalakis2009complexity-PPAD,chen2006settling-ppad}. Beyond computation, \citet{conitzer2008new-np} showed that even deciding whether a strong NE exists is {\sf NP}-complete in two-player symmetric games, and \citet{berthelsen2022computational-exist-R} further established that the problem is $\mathsf{\exists\RR}$-complete for three-player games. Similarly, \citet{rahn2015efficient-collaborative-game} proved that determining the existence of a strong NE is {\sf NP}-complete in graph coordination games, even when restricting attention to coalitions of constant size. To the best of our knowledge, however, no hardness results are known for computing strong equilibria when correlation on the joint strategy is allowed, \emph{i.e.}, a correlated strategy immune to coalition deviations. In this paper, we show that computing a correlated strategy that minimizes the average gain from coalition deviations is {\sf NP}-hard. Moreover, we establish a lower bound based on the treewidth of the \Cref{def:utility-dependence-graph}, demonstrating an inherent computational barrier in solving the MASE considered here.

\paragraph{Computation of Strong Equilibrium.} \citet{holzman1997strong-congestion} and \citet{rozenfeld2006strong-congestion} developed algorithms to compute strong NE and correlated strong equilibria in congestion games under certain conditions in polynomial time. \citet{rahn2015efficient-collaborative-game} showed that a strong NE can also be computed in polynomial time when the graph coordination game is defined on a tree. In contrast, \citet{gatti2017verification-compute-SNE} proposed a spatial branch-and-bound algorithm for computing strong NE more generally, but its runtime is exponential. Along the same lines, \citet{nessah2014existence-SNE} also provided a computationally intractable algorithm. Of independent interest, \citet{papadimitriou2008computing-against-hope} introduced an efficient algorithm for computing optimal CE, \emph{e.g.}, a CE that maximizes the social welfare, in graphical games \citep{DBLP:conf/uai/KearnsLS01-graphical-game,kakade2003correlated-graphical-game} with bounded treewidth, using linear programming. In this paper, we develop a new algorithm for computing MASE based on no-regret learning, with a time complexity that matches the lower bound dictated by the treewidth of the \Cref{def:utility-dependence-graph}.

\section{Preliminaries}

For any vector $\bx \in \RR^n$, we use $x_i$ to denote its $i^{th}$ element and $\nbr{\bx}_p$ to denote its $p$-norm. By default, $\nbr{\bx}$ refers to the 2-norm. For a positive integer $N$, let $[N] \coloneqq {1,2,\dots,N}$. We denote the $(n-1)$-dimensional probability simplex by $\Delta^n\coloneqq \cbr{\bx\in[0,1]^n\colon \sum_{i=1}^n x_i=1}$. More generally, for any discrete set $S$, let $\abr{S}$ denote its cardinality and write $\Delta^S$ for the probability simplex over $S$, whose coordinates are indexed by elements of $S$ (\emph{e.g.}, $\Delta^n = \Delta^{\sbr{n}}$). Similarly, $\RR^S$ denotes the $\abr{S}$-dimensional real vector space with coordinates indexed by $S$. For any set $S_1, S_2$, $S_1 \times S_2$ denotes the Cartesian product of sets $S_1$ and $S_2$. Finally, we let $\ind(\text{argument})$ denote the indicator function, which equals $1$ if the argument is true and $0$ otherwise.

\subsection{Games}

A game is represented as a tuple $(N, \cbr{\cA_i}_{i=1}^N, \cbr{\cU_i}_{i=1}^N, \cS)$, where
\begin{itemize}[noitemsep, topsep=0pt]
    \item $N$ is the number of players.
    \item $\cA_i$ is the action set of player $i$. For convenience, let $\cA \coloneqq \bigtimes_{i=1}^N \cA_i$ denote the joint action set.%
    \item $\cU_i \colon \cA \to [0,1]$ is the utility function of player $i\in[N]$.
    \item $\cS$ is the set of coalitions, which is a set of subsets of players. For example, if only unilateral deviations are allowed (as in Nash equilibrium or coarse correlated equilibrium), then $\cS = \cbr{\cbr{1}, \cbr{2}, \dots, \cbr{N}}$.
\end{itemize}
For notational simplicity, for any subset of players $S \subseteq [N]$, we write $\cA_S \coloneqq \bigtimes_{i \in S} \cA_i$. Throughout the paper, let $A\coloneqq\max_{i\in [N]} \abr{\cA_i}$ denote the size of the largest action set.

For any joint action $\ba \in \cA$, let $a_i$ denote the action of player $i$, and let $\ba_{-i}=(a_1,a_2,\dots,a_{i-1},a_{i+1},\dots a_N)$ be the joint action of all players except $i$. More generally, for any subset $S \subseteq [N]$, we write $\ba_{-S}$ for the joint action of players outside $S$.

\subsection{Succinct Representation}

This paper focuses on multi-player games with a succinct representation. Specifically, each utility function $\cU_i$ can be encoded using a number of bits polynomial in the number of players $N$, rather than requiring $\cO\rbr{N\prod_{i=1}^N |\cA_i|}$ bits as in the general case. Examples of succinctly represented games include polymatrix games \citep{howson1972equilibria-polymatrix,eaves1973polymatrix} and congestion games \citep{rosenthal1973class-congestion}. Throughout the paper, we call an algorithm \emph{efficient} if its running time is polynomial in $N$, as opposed to polynomial in $\prod_{i=1}^N |\cA_i|$. We focus on succinct games because MASE can otherwise be solved by a linear program whose size grows exponentially with $N$ (see \Cref{section:lp}). Moreover, the study of strong equilibrium is particularly compelling in large games, where exponential dependence on $N$ is computationally prohibitive.

\section{Minimum Average-Strong Equilibrium (MASE)}

Several notions of strong equilibrium have been proposed, including the strong Nash equilibrium (NE) \citep{strong-NE}, the sum-strong NE \citep{hoefer2013strategic-equivalence-CSE-CSSE} (no improvement on the total gain of any coalition), and coalition-proof equilibrium \citep{bernheim1987coalition-proof} (permits only self-enforcing deviations, see \Cref{section:relation}). However, none of these notions are guaranteed to exist in general games. To build intuition, we first illustrate why a strong NE does not exist in Prisoner’s Dilemma. Later, we will show that the problem persists even when correlated strategies are allowed.

\begin{table}[t]
    \centering
    \begin{tabular}{|c|c|c|}
    \hline
     & \textbf{Confess (C)} & \textbf{Defect (D)} \\
    \hline
    \textbf{Confess (C)} & $(0.6,0.6)$ & $(0,1)$ \\
    \hline
    \textbf{Defect (D)} & $(1,0)$ & $(0.2,0.2)$ \\
    \hline
    \end{tabular}
    \caption{Utility matrix of Prisoner’s Dilemma. Each entry $(a,b)$ denotes the payoff of the row player ($a$) and the column player ($b$).}
    \label{table:PD}
\end{table}
\begin{lemma}
\label{lemma:non-existence}
    In Prisoner's Dilemma, no strong Nash nor strong correlated equilibrium exists when $\cS=\cbr{\cbr{1},\cbr{2},\cbr{1,2}}$ is the set of all non-empty subsets of players.
\end{lemma}
\begin{proof}
Since correlated equilibria include all Nash equilibria, it suffices to examine strong correlated equilibria. A strong correlated equilibrium is a correlated joint strategy where no subset of players (a coalition) can jointly deviate in a way that strictly improves the utility of all its members.

As shown in \Cref{table:PD}, any strategy with positive weight on $(C,C), (C,D), (D,C)$ yields a profitable deviation for at least one singleton coalition, $\cbr{1}$ or $\cbr{2}$. Conversely, placing all weight on $(D,D)$ creates a deviation to $(C,C)$ that benefits the coalition $\cbr{1,2}$. Thus, no strong NE or strong correlated equilibrium exists.\qedhere
\end{proof}
In \Cref{lemma:non-existence-coalition-proof}, we further show that even the weaker notion, the coalition-proof equilibrium, may fail to exist in normal-form games, even when correlation is allowed.

The failure of strong equilibria stems from the fact that a single player may belong to multiple coalitions whose objectives are incompatible, so no strategy can simultaneously eliminate profitable deviations for all coalitions. Motivated by this, rather than requiring no profitable deviation exists, we instead look for a joint strategy that \emph{minimizes} the coalitional incentive to deviate, which always exists by the Weierstrass theorem. 

For each coalition $S\in\cS$, we adopt a transferable-utility viewpoint: after a deviation, members of $S$ may redistribute utility internally. Under this assumption, a coalition's incentive to deviate is naturally summarized by its aggregate gain. To make this notion comparable across coalitions of different sizes, we evaluate it on a per-capita basis by averaging over the players in $S$. This motivates the following definition of the \emph{Minimum Average-Strong Equilibrium} (MASE). Extensions to the weighted sum (average) and to the maximum within-coalition gain are discussed in \Cref{remark:extension}. By contrast, extending the criterion to the minimum within-coalition gain is computationally intractable even in small games, as shown in \Cref{theorem:strong-CCE-hardness}. Let $\Pi\subseteq \Delta^{\cA}$ be the set of joint strategies under consideration. Then a MASE $\pi^*\in\Pi$ is any solution to
\begin{align}
    \pi^*\in \argmin_{\pi\in \Pi}{\color{red}\max_{S\in \cS} \max_{\hat\ba_S\in\cA_S} \frac{1}{|S|}\sum_{i\in S} \EE_{\ba\sim \pi}\sbr{\cU_i\rbr{\hat\ba_S, \ba_{-S}} - \cU_i\rbr{\ba} }}, \numberthis[MASE]{eq:MASE}
\end{align}
where the {\color{red} expression} above is referred to as the \emph{coalition exploitability} of $\pi$. To make the solution concept computationally tractable, and to allow coordinated behavior, we work with correlated joint strategies, rather than restricting attention to independent mixed strategies as in NE. Accordingly, throughout the remainder of the paper we take $\Pi=\Delta^{\cA}$. This choice helps avoid the computational hardness inherited from NE: when the coalition family $\cS$ contains only singleton coalitions, any strong-equilibrium notion reduces to the usual unilateral-deviation requirement, \emph{i.e.}, it collapses to NE. Since computing an NE is {\sf PPAD}-hard even for two-player normal-form games \citep{chen2006settling-ppad,daskalakis2009complexity-PPAD}, insisting on an equilibrium notion built on independent strategies would generally be intractable.

Intuitively, \Cref{eq:MASE} selects the correlated strategy $\pi \in \Delta^{\cA}$ that minimizes the maximum average gain attainable by any coalition across all possible coalitions. If this value is less than or equal to zero, then no coalition can simultaneously deviate in a way that yields a strictly positive total gain.

A correlated strategy $\pi \in \Delta^{\cA}$ is called an \emph{$\epsilon$-MASE} if, for every MASE $\pi^* \in \Delta^{\cA}$,
\begin{align*}
    &\max_{S\in \cS} \max_{\hat\ba_S\in\cA_S} \frac{1}{|S|}\sum_{i\in S} \EE_{\ba\sim \pi}\sbr{\cU_i\rbr{\hat\ba_S, \ba_{-S}} - \cU_i\rbr{\ba} }\\
    \leq& \max_{S\in \cS} \max_{\hat\ba_S\in\cA_S} \frac{1}{|S|}\sum_{i\in S} \EE_{\ba\sim {\color{orange!50!black}\pi^*}}\sbr{\cU_i\rbr{\hat\ba_S, \ba_{-S}} - \cU_i\rbr{\ba} } + \epsilon.
\end{align*}
A strong Nash equilibrium requires that for any deviating coalition, at least one member does not strictly improve their utility. In contrast, $\epsilon$-MASE aims to minimize the average improvement over all players within any given coalition. From another perspective, $\epsilon$-MASE minimizes the incentive to deviate, even when coalition members can freely reallocate utility within the coalition.

\definecolor{cOuter}{HTML}{8EC5C1}   %
\definecolor{cMidA}{HTML}{4B90B0}    %
\definecolor{cMidB}{HTML}{2F6B8A}    %
\definecolor{cInner}{HTML}{1D4F6A}   %

\section{Relations among Equilibria}
\label{section:relation}

In this section, we delineate the relationships among various strong equilibrium concepts, including strong CCE (the counterpart to strong NE \citep{strong-NE}),\footnote{Here, strong CCE refers to robustness against coalition deviations, analogous to strong NE. This distinguishes it from the definition in \citet{strong-CCE-different-def}, where "strong" implies that any unilateral deviation strictly decreases the deviating player's utility.} strictly strong CCE (the counterpart to strictly strong NE \citep{van1996strong-strictly-strong}), and coalition-proof CCE (the counterpart to coalition-proof NE \citep{bernheim1987coalition-proof}).

To align with the standard equilibrium requirement that no player or coalition can gain by deviating, we say that a joint strategy $\pi$ is a $\text{MASE}^{\leq 0}$ if it satisfies\footnote{An $\epsilon$-MASE is an $\epsilon$-approximate minimizer of the coalition exploitability, whereas $\text{MASE}^{\leq 0}$ is a joint strategy whose coalition exploitability is at most zero.}
\begin{align}
    \max_{S\in \cS} \max_{\hat\ba_S\in\cA_S} \frac{1}{|S|}\sum_{i\in S} \EE_{\ba\sim \pi}\sbr{\cU_i\rbr{\hat\ba_S, \ba_{-S}} - \cU_i\rbr{\ba} }\leq 0. \numberthis[$\text{MASE}^{\leq 0}$]{eq:def-MASE-leq-0}
\end{align}
This notion coincides with the sum-strong correlated equilibrium \citep{hoefer2013strategic-equivalence-CSE-CSSE} (no improvement on the total gain of any coalition), as the sign of the aggregate deviation gain is invariant to scaling by the coalition size (\emph{i.e.}, the average gain is non-positive if and only if the sum is non-positive).

We can unify these equilibrium notions under a single generalized definition. Note that while standard definitions assume $\cS$ contains all non-empty subsets of $[N]$ for strictly strong, strong, and coalition-proof CCE, we generalize $\cS$ in this work to represent an arbitrary collection of subsets to align with the MASE framework. The unified condition is given by:
\begin{align}
    \max_{S\in \cS} \max_{\hat\pi_S\in \cC_S} f\rbr{\cbr{\EE_{\ba\sim \pi, \hat\ba_S\sim \hat\pi_S}\sbr{\cU_i\rbr{\hat\ba_S, \ba_{-S}} - \cU_i\rbr{\ba} }}_{i\in S}} \leq 0,\label{eq:generalized-condition}
\end{align}
where
\begin{align}
    (\cC_S, f(G))\coloneqq \begin{cases}
        \rbr{\Delta^{\cA_S}, \frac{1}{|G|} \sum_{g\in G} g} & \text{MASE}^{\leq 0}\\
        \rbr{\Delta^{\cA_S}, \max_{g\in G} g - \ind\rbr{\min_{g'\in G} g'< 0}} & \text{strictly strong CCE}\\
        \rbr{\Delta^{\cA_S}, \min_{g\in G} g }& \text{strong CCE}\\
        \rbr{\text{self-enforcing strategies}, \min_{g\in G} g }& \text{coalition-proof CCE}
    \end{cases}
\end{align}
The definition of self-enforcing strategies is discussed later in this section. For $\text{MASE}^{\leq 0}$, maximizing over $\hat\pi_S \in \Delta^{\cA_S}$ in \Cref{eq:generalized-condition} is equivalent to maximizing directly over $\hat\ba_S \in \cA_S$ in \Cref{eq:MASE}. This is because the function $f$ corresponding to $\text{MASE}^{\leq 0}$ is linear in $\hat\pi_S$, so its maximum over the simplex is attained at an extreme point (vertex), \emph{i.e.}, at a deterministic choice $\hat\pi_S$ with $\hat\pi_S(\hat \ba_S)=1$ for some $\hat \ba_S\in\cA_S$.

For strictly strong CCE, since $\cU_i\in\sbr{0,1}$, each gain $g\in\sbr{-1,1}$, and thus $\max_{g\in G} g \leq 1$. Hence, we have
\begin{align}
    f(G)\leq 0\quad\text{if and only if}\quad \max_{g\in G} g\leq 0 \text{ or } \min_{g'\in G} g' < 0.
\end{align}
Indeed, if $\min_{g'\in G} g' < 0$, then $f(G)=\max_{g\in G} g-1\leq 0$ automatically. Otherwise, $\min_{g'\in G} g' \geq 0$ and the condition reduces to $\max_{g\in G} g\leq 0$, which forces $g\leq 0$ for all $g\in G$. Equivalently, $f(G)\leq 0$ fails exactly when $\min_{g'\in G} g' \geq 0$ and $\max_{g\in G} g > 0$, \emph{i.e.}, when a deviation weakly benefits every coalition member and strictly benefits at least one.

\paragraph{Distinction between strictly strong and strong CCE.} A strictly strong CCE rules out any coalition deviation that makes all members weakly better off and at least one member strictly better off. A strong CCE rules out any coalition deviation that makes every member strictly better off.

\paragraph{Distinction between coalition-proof and strong CCE.} The primary distinction between strong CCE and coalition-proof CCE lies in the set of admissible deviations. For strong CCE, any deviation within $\Delta^{\cA_S}$ is considered feasible.

In contrast, coalition-proof CCE restricts the set of feasible deviations to those that are self-enforcing. Specifically, a deviation is admissible only if it constitutes a coalition-proof equilibrium in the reduced game restricted to players in $S$, holding the strategies of players in $[N]\setminus S$ fixed at $\pi_{-S}$. Further details are provided in \citet{bernheim1987coalition-proof}. Consequently, since the set of deviations in coalition-proof CCE is a subset of that in strong CCE, any strong CCE is inherently a coalition-proof CCE. Relations between these strong equilibria are summarized in \Cref{fig:strong-equilibrium-relation}.

In \Cref{lemma:non-existence-coalition-proof}, we extend the non-existence result for (non-correlated) coalition equilibria in \citet{bernheim1987coalition-proof} to establish the non-existence of coalition-proof CCE, a class that includes all coalition equilibria.

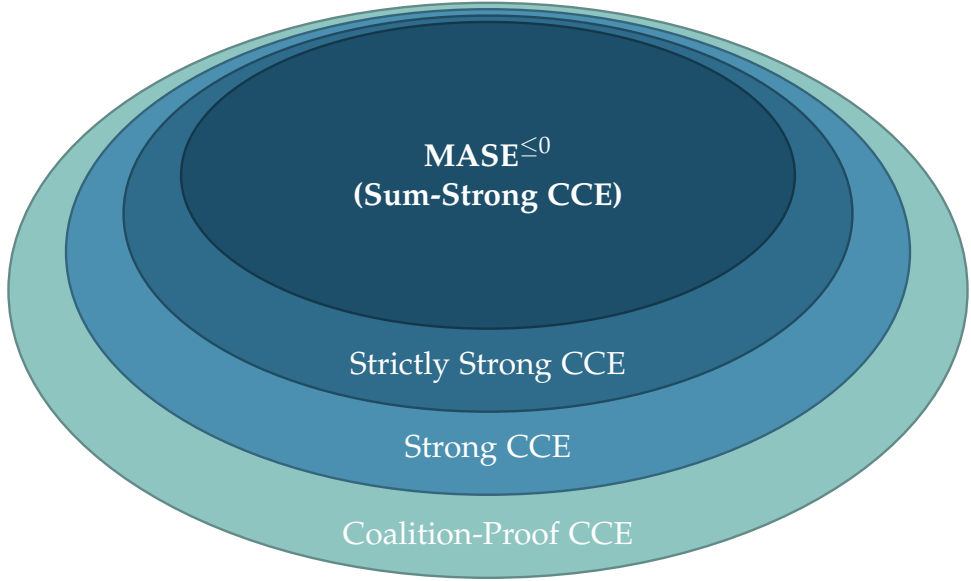
\begin{figure}
    \centering
\begin{tikzpicture}[line join=round, line cap=round]
  \def\sx{0.85}      %
  \def\sy{0.85}      %
  \newlength{\RxBase}\setlength{\RxBase}{7.5cm}
  \newlength{\RyBase}\setlength{\RyBase}{4.5cm}
  \newlength{\stepxBase}\setlength{\stepxBase}{0.9cm}
  \newlength{\stepyBase}\setlength{\stepyBase}{0.7cm}

  \pgfmathsetlengthmacro{\Rx}{\sx*\RxBase}
  \pgfmathsetlengthmacro{\Ry}{\sy*\RyBase}
  \pgfmathsetlengthmacro{\stepx}{\sx*\stepxBase}
  \pgfmathsetlengthmacro{\stepy}{\sy*\stepyBase}

  \coordinate (C0) at ( 0.0, 0.00); %
  \coordinate (C1) at ( 0.0, {0.60 * \sy}); %
  \coordinate (C2) at ( 0.0, {1.20 * \sy}); %
  \coordinate (C3) at ( 0.0, {1.80 * \sy}); %

  \newcommand{\OvalAt}[4]{%
    \draw[fill=#3, draw=#4!70!black, line width=0.9pt]
      #1 ellipse ({\Rx-#2*\stepx} and {\Ry-#2*\stepy});
  }

  \OvalAt{(C0)}{0}{cOuter}{cOuter} %
  \OvalAt{(C1)}{1}{cMidA}{cMidA}   %
  \OvalAt{(C2)}{2}{cMidB}{cMidB}   %
  \OvalAt{(C3)}{3}{cInner}{cInner} %

  \node[font=\large\bfseries, text=white, align=center] at (C3) {$\text{MASE}^{\leq 0}$\\(Sum-Strong CCE)};
  \node[font=\large, text=white] at ($(C0) + ( 0, {-\Ry+1*\stepy})$) {Coalition-Proof CCE};
  \node[font=\large, text=white] at ($(C1) + ( 0, {-\Ry+2*\stepy})$) {Strong CCE};
  \node[font=\large, text=white] at ($(C2) + ( 0, {-\Ry+3*\stepy})$) {Strictly Strong CCE};
\end{tikzpicture}
    \caption{The relationship between different strong equilibrium notions.}
    \label{fig:strong-equilibrium-relation}
\end{figure}

\begin{lemma}[Non-existence of Coalition-Proof CCE] 
\label{lemma:non-existence-coalition-proof}
Coalition-proof CCE may not exist, even in a three-player normal-form game.
\end{lemma}
The proof is postponed to \Cref{section:relation-proof}.

The non-existence possibility in \Cref{lemma:non-existence-coalition-proof} automatically implies the non-existence of $\text{MASE}^{\leq 0}$, strictly strong CCE, and strong CCE. This motivates our focus on minimizing the deviation gap rather than seeking a joint strategy that strictly achieves a non-positive deviation gap.

\section{Hardness of Computing MASE}
\label{section:hardness}

Recall that we call an algorithm \emph{efficient} if it runs in time polynomial in $N$. In this section, we first establish the computational hardness of computing $\epsilon$-MASE.

\begin{theorem}
\label{theorem:complexity-strong-equilibrium}
Computing $\epsilon$-MASE is {\sf NP}-hard, even when $\cS$ only contains singletons (coalitions of size one) and $1/\epsilon$ is polynomial in the number of players.
\end{theorem}

The proof is deferred to \Cref{section:proof-mase-hardness}. Importantly, \Cref{theorem:complexity-strong-equilibrium} highlights a fundamental distinction from the case of CCE, which can be computed efficiently \citep{papadimitriou2008computing-against-hope}. The reason is that for CCE it suffices to find a correlated strategy $\pi \in \Delta^{\cA}$ such that the deviation gap, $\max_{i\in [N]}\max_{\hat a_i\in\cA_i} \EE_{\ba\sim \pi}\sbr{\cU_i\rbr{\hat a_i, \ba_{-i}} - \cU_i\rbr{\ba} }$, is less or equal to zero, whereas here we must find a strategy that minimizes the gap. Together with the linear programming characterization in \Cref{section:lp}, this implies that computing $\epsilon$-MASE is actually {\sf NP}-complete. In fact, \citet{anagnostides2025polynomial-minty} recently showed that even minimizing the \emph{average} deviation gap of CCE across all players (instead of the maximum gap considered here) is also {\sf NP}-complete.

\subsection{Lower Bound on the Exponential Dependence on Treewidth}

Next, we present a more refined hardness result: a lower bound for computing MASE. To do so, we first formalize the notion of dependencies among players’ utilities.

For each player $i \in [N]$, define the relevant set $\cN(i)\subseteq [N]$ consisting of all players $j \in [N]$ (including $j=i$) such that the action of $j$ can affect the utility of $i$. Formally, $j\in [N]$ is in $\cN(i)$ if and only if there exist $\ba_{-j}\in \cA_{-j}$ and $a_j,a_j'\in\cA_j$ such that $\cU_i(a_j, \ba_{-j})\not= \cU_i(a_j', \ba_{-j})$. This leads to the following graph representation.

\begin{definition}[Utility Dependency Graph]
\begin{enumerate}[left=0mm,nosep,leftmargin=*, align=left, labelsep=0pt, labelwidth=0pt,label={},ref={Utility Dependency Graph}] \item The utility dependence graph $\cG=(\cV,\cE)$ is an undirected graph with vertex set $\cV=[N]$ representing the players, and edge set $\cE=\bigcup_{k\in [N]} \cbr{(i,j)\given i,j \in \cN(k),i\not=j}$. \label{def:utility-dependence-graph} 
\end{enumerate}
\end{definition}
Since $\cU_i$ depends only on the actions of players in $\cN(i)$, we may equivalently write $\cU_i(\ba_C)=\cU_i(\ba_C,\ba_{-C}')$ for arbitrary $\ba_{-C}'\in\cA_{-C}$, where $C\supseteq \cN(i)$. It is worth noting that this definition differs from the graph of a graphical game \citep{kakade2003correlated-graphical-game,DBLP:conf/uai/KearnsLS01-graphical-game}. Here, players $i$ and $j$ are connected if both influence the utility of some other player $k$, even if $i$ and $j$ do not directly affect each other. Whereas in graphical games, two players $i$ and $j$ are connected if and only if at least one can influence the other’s utility.

With this graph structure in place, we can connect the hardness of computing MASE to the treewidth of $\cG$. Intuitively, treewidth measures how close a graph is to being a tree: the treewidth of $\cG$ is one when $\cG$ is a tree, and it is $N-1$ when $\cG$ is a complete graph. Throughout this section, let $\cO^*$ denote asymptotic complexity with factors polynomial in $N$ suppressed.

\begin{theorem}[Treewidth]
\label{theorem:tree-width-important}
    Suppose a tree decomposition of the \Cref{def:utility-dependence-graph} is given. Under the Strong Exponential Time Hypothesis (SETH) \citep{impagliazzo2001complexity-SETH},\footnote{SETH assumes that SAT cannot be solved in $\cO^*((2-\zeta)^n)$ for any $\zeta>0$, where $n$ is the number of variables in the SAT instance.} \Cref{eq:MASE} cannot be computed in $\cO^*((A-\zeta)^{\tw{\cG}})$ for any $\zeta>0$. Moreover, under the additional assumption that {\sf BPP=P},\footnote{This assumption implies that any problem with a polynomial-time randomized algorithm also has a polynomial-time deterministic algorithm.} $\frac{1}{9N^2}$-approximate MASE cannot be computed in $\cO^*((A-\zeta)^{\tw{\cG}})$ for any $\zeta>0$.
\end{theorem}
The proof is deferred to \Cref{section:tree-width-proof}. For approximate MASEs we assume {\sf BPP=P}, which is standard in the literature \citep{arora2009computational-complexity-BPP-P}, because the reduction involves sampling joint actions from the approximate MASE. Since enumerating all $\ba \in \cA$ is computationally infeasible, we rely on randomized sampling. This yields only a randomized algorithm for the original {\sf NP}-hard problem, and the assumption {\sf BPP=P} ensures that such a randomized algorithm can be derandomized into a deterministic one, completing the reduction. Finally, we note that the argument can be extended to the maximum-within-coalition variant, showing that its computation also requires exponential dependence on the treewidth of the \Cref{def:utility-dependence-graph}.

\Cref{theorem:tree-width-important} shows that the computational complexity of solving MASE is inherently tied to the treewidth of the \Cref{def:utility-dependence-graph}. Intuitively, when the treewidth is large, each player’s utility depends on many others, making even the evaluation of coalition deviations computationally demanding (enumerating over all $\hat\ba_S\in\cA_S$). In contrast, when the treewidth is small, such as zero (each player’s utility depends only on their own action), computing MASE becomes trivial, since each player’s utility can be maximized independently. In polymatrix games \citep{eaves1973polymatrix}, the treewidth of the \Cref{def:utility-dependence-graph} can be bounded by that of its corresponding graph. Further details are provided in \Cref{section:polymatrix}.

Furthermore, we show that the minimum-value analogue of strong CCE is {\sf NP}-hard to approximate, even in small games whose joint action space is only polynomial in the number of players and the treewidth of the \Cref{def:utility-dependence-graph} is two.
\begin{theorem}
\label{theorem:strong-CCE-hardness}
Approximating
\begin{align*}
    \min_{\pi\in \Delta^{\cA}}\max_{S\in \cS} \max_{\hat\pi_S\in\Delta^{\cA_S}} \min_{i\in S}
    \EE_{\ba\sim \pi, \hat\ba_S\sim \hat\pi_S}
    \sbr{\cU_i\rbr{\hat\ba_S, \ba_{-S}} - \cU_i\rbr{\ba}}
    \numberthis[Minimum Strong CCE]{eq:strong-CCE}
\end{align*}
within an additive error of $\frac{1}{N^2}$ is {\sf NP}-hard, even for a small game in which the joint action set size $|\cA|$ is polynomial in the number of players $N$ and the treewidth of the \Cref{def:utility-dependence-graph} is two.
\end{theorem}

The proof is deferred to \Cref{section:strong-CCE-hardness-proof}. \Cref{theorem:strong-CCE-hardness} shows that computing \Cref{eq:strong-CCE} is intractable unless ${\sf P} = {\sf NP}$. Accordingly, throughout the remainder of the paper, we focus on aggregation rules based on the (weighted) average and the maximum deviation gain.

\begin{remark}
    The hard instance in \Cref{theorem:complexity-strong-equilibrium} has a joint action set whose size is exponential in the number of players, whereas the instance in \Cref{theorem:strong-CCE-hardness} does not. This distinction is important: by the linear-programming formulation in \Cref{section:lp}, MASE can be computed in time polynomial in $|\cA|$ for any finite game. In contrast, \Cref{theorem:strong-CCE-hardness} rules out even an approximation algorithm whose running time is polynomial in $|\cA|$ for \Cref{eq:strong-CCE}, unless ${\sf P}={\sf NP}$. Therefore, this hardness cannot be attributed to the size of the joint action space. Rather, it stems from the intrinsic min-max-min structure of the minimum strong CCE objective.
\end{remark}

\section{Efficient Computation of MASE}
\label{section:computation-MASE}

Although an \Cref{eq:MASE} lives in an exponentially large space (of size $|\cA|$), it can still be computed efficiently. This is because the equilibrium always admits a compact representation.
\begin{theorem}[Efficient Representation]
\label{theorem:existence}
    For any $\epsilon\geq 0$, at least one of the $\epsilon$-MASE can be represented as a linear combination of $\sum_{S\in\cS} |S| \cdot A^{\tw{\cG}}+1$ pure strategies, where $\tw{\cG}$ is the treewidth of \Cref{def:utility-dependence-graph}.
\end{theorem}
The proof is deferred to \Cref{section:proof-existence}. Intuitively, \Cref{theorem:existence} shows that there must be an $\epsilon$-MASE that always has a sparse representation. Since each pure strategy can be encoded by the corresponding joint action rather than by a full vector in $\Delta^{\cA}$, this sparsity provides a structural explanation for why efficient computation is possible. We emphasize, however, that \Cref{theorem:existence} is not used as a subroutine in the algorithm developed below. Rather, it serves as intuition for the compactness that the algorithm exploits.

\subsection{Meta-Game between the Correlator and Deviator}

To compute an \Cref{eq:MASE}, we reformulate the problem as a \emph{meta-game} between two players: the \emph{correlator} and the \emph{deviator} \citep{hart1989existence-phi-zero-sum}. The correlator chooses the correlated strategy $\pi \in \Delta^{\cA}$, while the deviator selects deviations. The game is zero-sum: the correlator aims to minimize the coalition’s gain from deviation, and the deviator aims to maximize it. Formally:
\begin{align}
    \min_{\pi\in\Delta^{\cA}}~\max_{\mu\in\Delta^{\bigtimes_{S\in\cS} \cA_S}} ~ F(\pi,\mu), \label{eq:def-sum}
\end{align}
where
\begin{align}
    F(\pi,\mu)\coloneqq \sum_{S\in\cS} \sum_{\hat \ba_S\in\cA_S} \frac{\mu(S,\hat \ba_S) }{|S|}\sum_{i\in S} \EE_{\ba\sim \pi}\sbr{\cU_i\rbr{\hat \ba_S, \ba_{-S}} - \cU_i\rbr{\ba}}.
\end{align}
Here, we extend the deviator's decision space from a discrete to a continuous set. This relaxation does not strengthen the deviator, since the objective is linear in $\mu$, and the maximum is always attained at an extreme point. Therefore, \Cref{eq:def-sum} is equivalent to the original definition in \Cref{eq:MASE}.

A natural idea is to apply no-regret learning algorithms simultaneously for the correlator and deviator. However, directly updating the full distributions $\pi$ and $\mu$ is infeasible, because the underlying spaces are exponentially large.

Fortunately, \Cref{theorem:existence} implies that maintaining the full distributions is unnecessary: it suffices to keep track of a polynomial number of pure strategies, and use their convex combination as the approximate equilibrium. This motivates our use of \emph{Follow the Perturbed Leader} (FTPL) \citep{hazan2016introduction}, where each decision at a timestep is a pure strategy, which can be represented compactly.

Let $\pi^{(t)} \in \Delta^{\cA}$ and $\mu^{(t)} \in \Delta^{\bigtimes_{S \in \cS} \cA_S}$ denote the decision variables at timestep $t \geq 1$ for the correlator and the deviator, respectively. The interaction between these two players can be described by the update rule
\begin{equation}
\begin{split}
    &\pi^{(t+1)}\in \argmin_{\pi\in\Delta^{\cA}} \sum_{\tau=1}^t F\rbr{\pi,\mu^{(\tau)}} - \inner{\tilde\bn^{(t+1)}}{\pi} \\
    &\mu^{(t+1)}\in \argmax_{\mu^{(t)}\in \Delta^{\bigtimes_{S\in\cS} \cA_S}}~ \sum_{\tau=1}^t F\rbr{\pi^{(\tau)},\mu} + \inner{\tilde\bbm^{(t+1)}}{\mu},
\end{split}
\label{eq:update-rule}
\end{equation}
where $\tilde\bn^{(t+1)}$ and $\tilde\bbm^{(t+1)}$ are noise vectors sampled independently at each timestep from some distribution, which we will specify later. These noise terms play the role of regularizers in online mirror descent (OMD) \citep{hazan2016introduction}, ensuring stability in the updates by controlling $\EE\sbr{\nbr{\pi^{(t+1)} - \pi^{(t)}}}$ and $ \EE\sbr{\nbr{\mu^{(t+1)} - \mu^{(t)}}}$.

Since $F(\pi,\mu)$ is bilinear in $(\pi, \mu)$, both the minimization and maximization problems admit solutions at vertices of their respective decision spaces. In other words, the $\argmin$ for the correlator and the $\argmax$ for the deviator always contain at least one pure strategy.

In what follows, we will explain in detail how to update $\pi$ efficiently under this framework. The update of $\mu$ is deferred to \Cref{section:mu-update}.

\subsection{Efficient Update of $\pi$}

The key step in updating $\pi^{(t+1)}$ is to select a pure strategy, \emph{i.e.}, a joint action $\ba^{(t+1)} \in \cA$ with $\pi^{(t+1)}(\ba^{(t+1)})=1$, that minimizes the objective. To gain insight into this update rule, we first examine how to compute $\argmin_{\pi \in \Delta^{\cA}} F(\pi,\mu)$ for a fixed $\mu$.

Suppose we want to find a joint action $\tilde \ba \in \cA$ such that the pure strategy $\tilde \pi$ with $\tilde \pi(\tilde \ba)=1$ minimizes $F(\tilde \pi,\mu)$. Expanding the definition, we obtain
\begin{align*}
    F(\tilde \pi,\mu)=&\sum_{i=1}^N~ \sum_{S\in\cS\colon i\in S}~ \sum_{\hat \ba_S\in\cA_S} \frac{\mu(S,\hat \ba_S) }{|S|} \rbr{\cU_i\rbr{\hat \ba_S, \tilde \ba_{-S}} - \cU_i\rbr{\tilde \ba}}\\
    =&\sum_{i=1}^N~ {\color{red}\sum_{S\in\cS\colon i\in S}~ \sum_{\hat \ba_S\in\cA_S} \frac{\mu(S,\hat \ba_S) }{|S|} \rbr{\cU_i\rbr{\hat \ba_{S\cap \cN(i)}, \tilde \ba_{\cN(i)\setminus S}} - \cU_i\rbr{\tilde \ba_{\cN(i)}}}}.
\end{align*}
Therefore, for each candidate $\tilde \ba \in \cA$, only the local actions $\tilde \ba_{\cN(i)}$ matter for the {\color{red} expression} above. If we can evaluate this {\color{red} expression} efficiently,\footnote{This is possible since $\mu$ is a linear combination of pure strategies when updated according to \Cref{eq:update-rule}.} then for each player $i \in [N]$ we may search for $\tilde \ba_{\cN(i)}$ that minimizes it. However, a difficulty arises because $\cN(i)$ and $\cN(j)$ may overlap across different players. Hence, we must ensure that the local assignments remain globally consistent.

To address this, we now introduce the concept of a tree decomposition and show how it enables us to optimize $F$ efficiently. Throughout the paper, we assume that a tree decomposition is given, and analyze the complexity only with respect to this decomposition.

\begin{figure}
    \centering
    \includegraphics[width=0.95\linewidth]{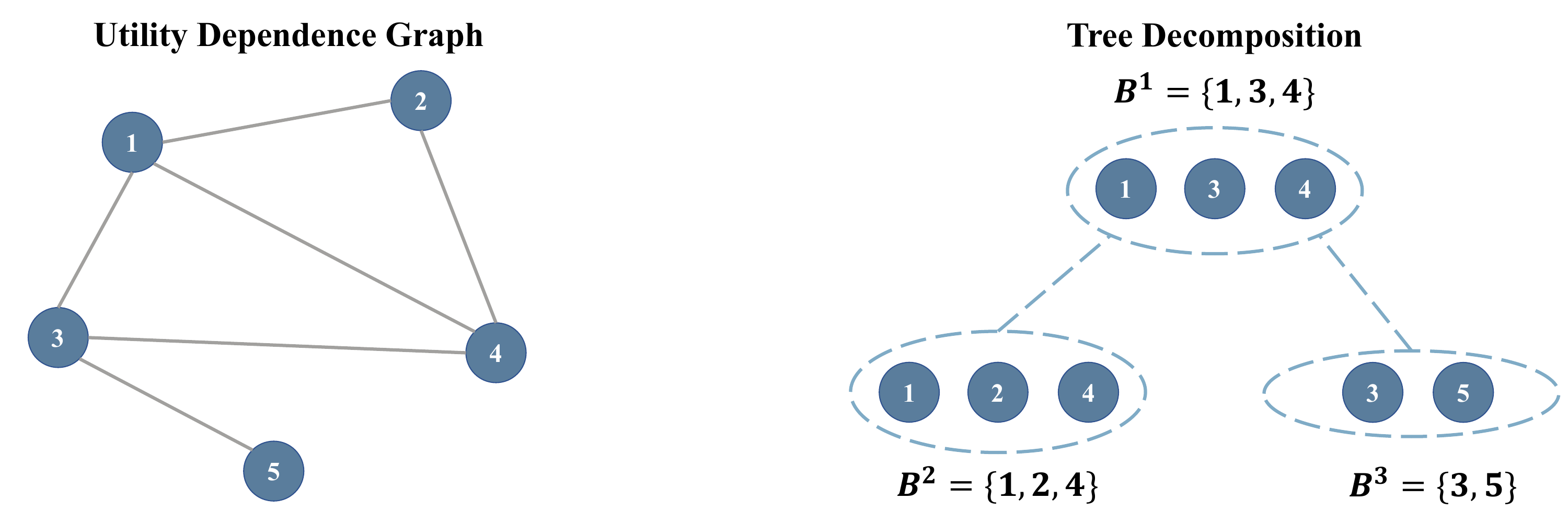}
    \caption{An illustration of a tree decomposition of the \Cref{def:utility-dependence-graph}.}
    \label{fig:tree-decomposition}
\end{figure}

\paragraph{Tree decomposition.}
A tree decomposition $\cT\coloneqq\cbr{B^1,B^2,\dots,B^K}$ of the \Cref{def:utility-dependence-graph} $\cG=(\cV, \cE)$ is a tree with $K$ nodes (bags), each $B^k \subseteq \cV$ where $\cV=[N]$, satisfying the following properties \citep{diestel2025graph}:
\begin{enumerate}[ref={\bf Property~\arabic* of Tree Decomposition},nosep]
    \item $\bigcup_{k=1}^K B^k = [N]$.\label{item:td-property1}
    \item For every edge $(i,j)\in \cE$, there exists $k$ with $\cbr{i,j}\subseteq B^k$.\label{item:td-property2}
    \item For any player $i\in [N]$, if $i$ appears in two bags $B,B'\in\cT$, then every bag on the path from $B$ to $B'$ also contains $i$.\label{item:td-property3}
\end{enumerate}
The treewidth of $\cG$ is
\begin{align}
    \tw{\cG}\coloneqq \min_{\cT}\max_{B\in\cT}\rbr{|B|-1},
\end{align}
that is, the minimum, over all tree decompositions, of the maximum bag size minus one. As illustrated in \Cref{fig:tree-decomposition}, the tree decomposition separates the game into overlapping bags. For example, since $B^2$ and $B^3$ only overlap at $B^1$, then $B^2$ and $B^3$ can be optimized independently, with consistency later enforced at $B^1$.

As illustrated in \Cref{fig:tree-decomposition}, a tree decomposition separates the game into overlapping local components. For example, if the subtrees rooted at $B^2$ and $B^3$ intersect only through $B^1$, then the assignments within these two subtrees can be optimized independently once their shared actions on $B^1$ are fixed; global consistency is subsequently enforced through the overlaps between adjacent bags.

Since any clique in $\cG$ is contained in some bag \citep{diestel2025graph}, for every player $i\in [N]$ there exists a bag $B$ with $\cN(i)\subseteq B$. We arbitrarily assign each player $i$ to such a bag.

\paragraph{Dynamic programming on the tree.}
We begin by choosing an arbitrary bag as the root of the tree decomposition and denote it by $B^r$. For each bag $B \in \cT$, let $C(B)$ denote the set of its children. With this setup, we maintain a vector $\bd^{(t+1)} \in \RR^{\bigtimes_{B\in\cT} \cA_B}$, defined as
\begin{equation}
\begin{split}
    d^{(t+1)}(B, \ba_B)=&\sum_{\tau=1}^t \sum_{S\in\cS} \frac{1}{|S|}\sum_{\substack{i\in S\colon\\i\text{ assigned to }B}}{\color{blue}\sum_{\hat \ba_S\in\cA_S} }\mu^{(\tau)}\rbr{S, \hat \ba_S} \rbr{\cU_i\rbr{(\hat\ba_{S\cap B},\ba_{B\setminus S})} - \cU_i(\ba_B)}\\
    &+\sum_{B'\in C(B)} \min_{\substack{\ba_{B'}'\in\cA_{B'}\colon\\\ba_{B\cap B'}=\ba'_{B\cap B'}}} d^{(t+1)}(B', \ba_{B'}')-n^{(t+1)}(B, \ba_B),
\end{split}
\label{eq:pi-dp}
\end{equation}
where $n^{(t+1)}(B, \ba_B)\sim {\rm  Exp}\rbr{\eta}$\footnote{$\Pr(x\geq w)=\exp(-\eta w)$ when $x\sim \Exp\rbr{\eta}$.} is sampled from an exponential distribution. Therefore, in \Cref{eq:update-rule}, $\tilde n^{(t+1)}(\ba)=\sum_{B\in\cT} n^{(t+1)}(B, \ba_B)$. 
Since each $i$ assigned to $B$ satisfies $\cN(i)\subseteq B$, the utility $\cU_i(\ba_B)$ can be written in terms of $\ba_B$ alone. Moreover, the summation ${\color{blue} \sum_{\hat \ba_S \in \cA_S}}$ can be computed efficiently, since $\mu^{(\tau)}$ is updated via \Cref{eq:update-rule} and is therefore a pure strategy.

\paragraph{Reconstructing the strategy.}
The optimal joint action $\ba^{(t+1)}\in\cA$ is then reconstructed recursively from the root $B^r$ to the leaves:
\begin{equation}
\begin{split}
    &\ba_{B^r}^{(t+1)}=\argmin_{\ba_{B^r}\in \cA_{B^r}} d^{(t+1)}(B^r, \ba_{B^r})\\
    \forall B\in C(B^r),~~~& \ba_{B\setminus B^r}^{(t+1)}=\argmin_{{\color{green!50!black}\ba_{B\setminus B^r}\in\cA_{B\setminus B^r}}} d^{(t+1)}\rbr{B, (\ba_{B\setminus B^r}, \ba^{(t+1)}_{B\cap B^r})}.
\end{split}
\label{eq:find-argmax-dp-pi}
\end{equation}
By \Cref{item:td-property1}, every player's action will be included. Since $\argmin_{{\color{green!50!black}\ba_{B\setminus B^r}\in\cA_{B\setminus B^r}}}$ is taken over $\cA_{B\setminus B^r}$, no contradictions arise by \Cref{item:td-property3}. We then set $\pi^{(t+1)}(\ba^{(t+1)})=1$.

The regret bound of this procedure is summarized below.
\begin{theorem}
\label{theorem:pi-no-regret}
Consider \Cref{eq:update-rule}. For any $\delta>0$, with probability at least $1-\delta$, the following holds:
\begin{align*}
    \max_{\hat\pi\in\Delta^{\cA}} \sum_{t=1}^T F\rbr{\pi^{(t)},\mu^{(t)}} - F\rbr{\hat\pi,\mu^{(t)}}\leq 2\abr{\cT} \frac{1+\rbr{\tw{\cG}+1}\log A}{\eta} + 2\eta\abr{\cT} T + \sqrt{2T\log \frac{1}{\delta}}.
\end{align*}
\end{theorem}
The proof is given in \Cref{section:pi-no-regret-proof}. Importantly, \Cref{theorem:pi-no-regret} shows that by setting $\eta=1/\sqrt{T}$, we obtain $\cO(\sqrt{T})$ regret. Since the update rule for $\mu$ mirrors that of $\pi$, the detailed analysis is deferred to \Cref{section:mu-update}. We now formally state the regret bound for $\mu$ in the following theorem.
\begin{theorem}
\label{theorem:mu-no-regret}
Consider the updates in \Cref{eq:update-rule}. For any $\delta>0$, with probability at least $1-\delta$, the following holds:
\begin{align*}
    \max_{\hat\mu\in\Delta^{\bigtimes_{S\in\cS} \cA_S}} \sum_{t=1}^T F\rbr{\pi^{(t)},\hat\mu} - F\rbr{\pi^{(t)},\mu^{(t)}}\leq& 2\abr{\cT} \frac{1+\rbr{\tw{\cG}+1}\log A}{\eta} + 2\eta\abr{\cT} T + \sqrt{2T\log \frac{1}{\delta}}.
\end{align*}
\end{theorem}
The complete proof is provided in \Cref{section:mu-update}.

\subsection{Computation of Equilibrium}

For any $\delta' > 0$, by setting $\delta = \frac{\delta'}{2}$ in \Cref{theorem:pi-no-regret} and \Cref{theorem:mu-no-regret}, and applying the union bound, we obtain that with probability at least $1 - \delta'$, the following holds:
\begin{equation}
\begin{split}
    &\max_{\hat\mu\in\Delta^{\bigtimes_{S\in\cS} \cA_S}} \sum_{t=1}^T F\rbr{\pi^{(t)},\mu} - \min_{\hat\pi\in\Delta^{\cA}}\sum_{t=1}^T F\rbr{\hat\pi,\mu^{(t)}}\\
    \leq& 4\abr{\cT} \frac{1+\rbr{\tw{\cG}+1}\log A}{\eta} + 4\eta\abr{\cT} T + 2\sqrt{2T\log\frac{2}{\delta'}}.
\end{split}
\label{eq:regret-sum}
\end{equation}
We now connect this bound to the convergence of the average strategy profile. Let $\pi^*, \mu^*$ be the solution to \Cref{eq:def-sum}, and define the average strategies $\overbar\pi \coloneqq \frac{1}{T} \sum_{t=1}^T \pi^{(t)}$ and $\overbar\mu \coloneqq \frac{1}{T} \sum_{t=1}^T \mu^{(t)}$. The left-hand side of \Cref{eq:regret-sum} corresponds to the duality gap: $\max_{\hat\mu\in\Delta^{\bigtimes_{S\in\cS} \cA_S}} F(\overbar\pi,\hat\mu) - \min_{\hat\pi\in\Delta^{\cA}} F(\hat\pi,\overbar\mu)$. Since $\pi^*, \mu^*$ are optimal solutions to \Cref{eq:regret-sum}, they satisfy
\begin{align*}
    \min_{\hat\pi\in\Delta^{\cA}} F(\hat\pi,\overbar\mu)\leq F(\pi^*,\mu^*)\leq \max_{\hat\mu\in\Delta^{\bigtimes_{S\in\cS} \cA_S}} F(\overbar\pi,\hat\mu).
\end{align*}
Combining these pieces, we arrive at the following finite-time convergence guarantee:
\begin{theorem}
    Let $\pi^*, \mu^*$ be the solution of \Cref{eq:def-sum}, and define $\overbar\pi \coloneqq \frac{1}{T}\sum_{t=1}^T \pi^{(t)},~\overbar\mu \coloneqq \frac{1}{T}\sum_{t=1}^T \mu^{(t)}$. Then, for any $\delta > 0$, with probability at least $1 - \delta$, we have
    \begin{align*}
        \max_{\hat\mu\in\Delta^{\bigtimes_{S\in\cS} \cA_S}} F(\overbar\pi,\hat\mu)\leq F(\pi^*,\mu^*) + 4\abr{\cT} \frac{1+\rbr{\tw{\cG}+1}\log A}{\eta T} + 4\eta\abr{\cT} + 2\sqrt{\frac{2\log\frac{2}{\delta}}{T}}.
    \end{align*}
\end{theorem}
With $\eta = \frac{1}{\sqrt{T}}$, the average strategy $\overbar\pi$ constitutes an $\cO\rbr{\frac{|\cT|\cdot \tw{\cG}\log A + \sqrt{\log \frac{2}{\delta}}}{\sqrt T}}$-MASE. The overall running time is $\cO\rbr{T\cdot|\cS|\cdot|\cT|\cdot A^{\tw{\cG}+1}}$. Hence, the exponential dependence aligns with the lower bound in \Cref{theorem:tree-width-important}.

\begin{remark}[Extension to Weighted Sum and Maximum over Coalition Gains]
\label{remark:extension}
The correlator-deviator framework in this section extends directly to other coalition objectives, including a weighted sum (or weighted average) of deviation gains and the maximum deviation gain within a coalition. In particular, we can optimize
\begin{align}
    &\min_{\pi\in \Delta^{\cA}}\max_{S\in \cS}\max_{\hat\ba_S\in\cA_S}
    \sum_{i\in S} w_{S,i}\,\EE_{\ba\sim \pi}\sbr{\cU_i\rbr{\hat\ba_S,\ba_{-S}}-\cU_i\rbr{\ba}}
    \numberthis[\textsc{Weighted Sum}]{eq:weighted-sum-objective}\\
    &\min_{\pi\in \Delta^{\cA}}\max_{S\in \cS}\max_{\hat\ba_S\in\cA_S}\max_{i\in S}
    \EE_{\ba\sim \pi}\sbr{\cU_i\rbr{\hat\ba_S,\ba_{-S}}-\cU_i\rbr{\ba}}
    \numberthis[\textsc{Maximum}]{eq:maximum-objective}\,,
\end{align}
where $\bw:=\cbr{w_{S,i}}_{S\in\cS,\, i\in S}$ is an arbitrary collection of weights (\emph{e.g.}, $\sum_{i\in S} w_{S,i}=1$ if one prefers a weighted \emph{average} for each coalition $S$).
For \Cref{eq:weighted-sum-objective}, the algorithm and analysis are unchanged. We focus on the uniform average solely for ease of exposition.
For \Cref{eq:maximum-objective}, it suffices to augment the deviator's decision to also select the player whose gain is evaluated, \emph{i.e.}, replace $\mu(S,\hat\ba_S)$ by $\mu(S,i,\hat\ba_S)$ over triples $(S,i,\hat\ba_S)$ with $i\in S$. The resulting meta-game is then equivalent to \Cref{eq:maximum-objective}.
\end{remark}

\section{Experiments}

In this section, we compare our algorithm against several baselines: Follow the Regularized Leader with a Euclidean regularizer (\textcolor[HTML]{ff9800}{FTRL}), \textcolor[HTML]{a31214}{Hedge}, Follow the Perturbed Leader with an exponential noise distribution (\textcolor[HTML]{b85c38}{FTPL}; all players run FTPL independently), and Online Mirror Descent with a Euclidean regularizer (\textcolor[HTML]{25597a}{OMD}) \citep{hazan2016introduction}. We also plot the ground-truth MASE computed via linear programming (\textcolor[HTML]{2e7d32}{LP}) in \Cref{section:lp}. Finally, to highlight the distinction between MASE and CCE, we also plot the CCE that minimizes the coalition exploitability defined below (\textcolor[HTML]{00897b}{Optimal CCE}).

We evaluate the algorithms on three criteria:
\begin{itemize}[nosep]
    \item \textbf{Exploitability.} ($\max_{i\in [N]} \max_{\hat a_i\in\cA_i} \EE_{\ba\sim \pi}\!\sbr{\cU_i(\hat a_i, \ba_{-i}) - \cU_i(\ba)}$): the maximum gain a single player can obtain by deviating unilaterally. Exploitability $\leq 0$ indicate a Nash equilibrium (or a correlated equilibrium if $\pi$ is correlated).
    \item \textbf{Coalition exploitability.} ($\max_{\mu\in\Delta^{\bigtimes_{S\in\cS}\cA_S}} F(\pi,\mu)$): the maximum average gain when a coalition deviates simultaneously. We take $\cS$ to be the set of all non-empty player subsets.
    \item \textbf{Social welfare.} ($\sum_{i=1}^N \EE_{\ba\sim\pi}\sbr{\cU_i(\ba)}$): the sum of all players' utilities.
\end{itemize}

Utility definitions and additional details are provided in \Cref{section:experiment-details}. In Prisoner's Dilemma \citep{prisoner-dilemma}, the MASE corresponds to players choosing $(C,D)$ and $(D,C)$ with probability $0.5$ each, yielding a social welfare of $1.0$. In contrast, because the unique NE/CCE in this game is $(D,D)$, the baselines converge to that outcome, with a lower social welfare of $0.4$. Thus, in Prisoner's Dilemma, MASE promotes cooperation and achieves higher utility.

In the Stag Hunt, there are two Nash equilibria, one of which attains higher utility. As shown in \Cref{fig:PD}, all baselines converge to the worse equilibrium, whereas MASE converges to the better one.
Finally, in terms of exploitability (unilateral deviations), MASE remains close to the baselines, while the baselines are substantially more fragile to multilateral deviations.

\begin{figure}[t]
    \centering
    \includegraphics[width=0.95\linewidth]{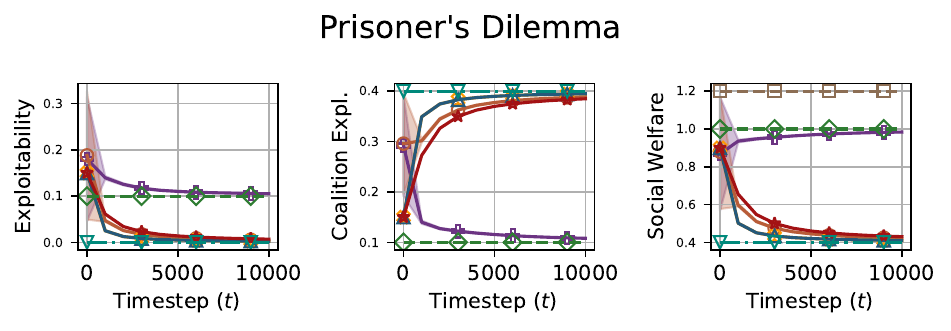}
    \includegraphics[width=0.95\linewidth]{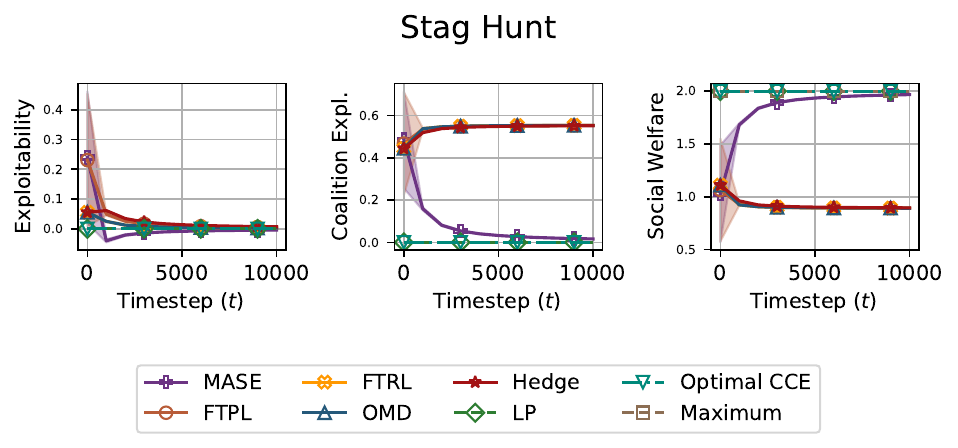}
    \caption{\textcolor[HTML]{2e7d32}{LP} denotes the linear programming solution from \Cref{section:lp}, and \textcolor[HTML]{8C6E54}{Maximum} denotes the maximum achievable social welfare. The baselines are comparatively fragile to multilateral deviations, while MASE is more robust and achieves higher social welfare. At the same time, MASE’s exploitability is close to that of the baselines.}
    \label{fig:PD}
\end{figure}

\subsection{Gap between Optimal CCE and Baseline Algorithms}

In \Cref{fig:optimal-CCE}, we compare the coalition exploitability achieved by the baseline algorithms, MASE, and the optimal CCE, \emph{i.e.}, the CCE with the smallest coalition exploitability. The figure shows that the optimal CCE remains close to MASE in terms of coalition exploitability, whereas none of the baseline algorithms approaches it. Instead, the baseline methods exhibit substantially larger coalition exploitability than MASE, highlighting the importance of explicitly accounting for multilateral deviations.

\begin{figure}[t]
    \centering
    \includegraphics[width=0.65\linewidth]{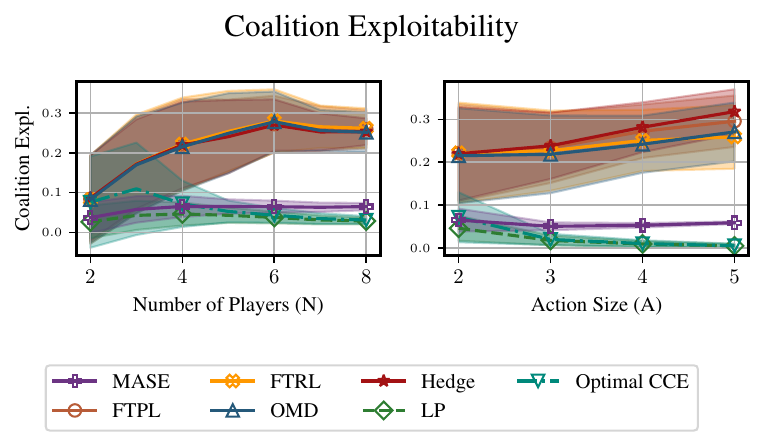}
    \caption{Coalition exploitability in random normal-form games with varying numbers of players and action-set sizes, where $\cS$ contains all possible coalitions ($|\cS|=2^N-1$). In each game, all players have the same action-set size $A$, and each entry of the utility matrix is sampled independently and uniformly from $[0,1]$, and then normalized so that the minimum and maximum entries are exactly $0$ and $1$, respectively.}
    \label{fig:optimal-CCE}
\end{figure}

\subsection{Coalition Exploitability in Larger Games}

As shown in \Cref{fig:polymatrix-results}, the coalition exploitability of the average strategy generated by classical no-regret learning algorithms increases as the game size grows. Note that we only consider coalitions of size no more than two. This trend underscores the importance of minimizing coalition exploitability. As games become larger, the equilibria to which these algorithms converge become increasingly fragile to coalition deviations, necessitating approaches that explicitly account for such multilateral deviations. Further details are provided in \Cref{section:polymatrix}.

\begin{figure}
    \centering
    \includegraphics[width=0.95\linewidth]{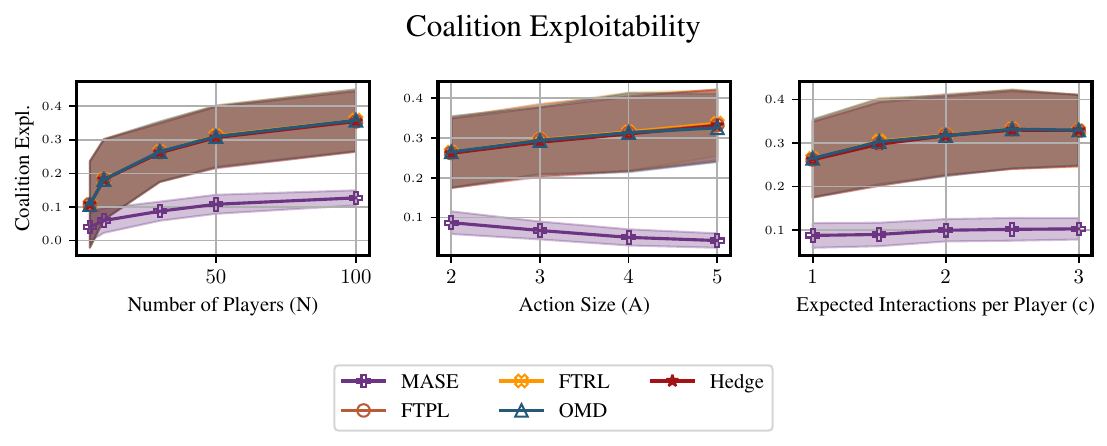}
    \caption{The coalition exploitability of random polymatrix games of different sizes when coalitions with no more than two players are considered. A larger expected number of interactions per player ($c$) generally corresponds to a larger treewidth of the \Cref{def:utility-dependence-graph}.}
    \label{fig:polymatrix-results}
\end{figure}

\section{Application: Tradeoff between Exploitability and Social Welfare}
\label{section:tradeoff}

In \Cref{fig:PD}, we see that allowing exploitability to increase from $0.0$ to $0.1$ raises the social welfare of MASE from $0.4$ to $1.0$. This begs a classical equilibrium-selection question, finding an equilibrium that maximizes social welfare \citep{conitzer2008new-np,papadimitriou2008computing-against-hope}, and its natural extension:
\begin{quote}
    Given a tolerance $\epsilon \ge 0$, what is the optimal $\epsilon$-approximate equilibrium that maximizes social welfare?
\end{quote}
In other words, if we allow an equilibrium to be exploitable, \emph{i.e.}, we permit unilateral deviations to improve a player’s expected utility by up to $\epsilon$, how much additional social welfare can we obtain? We formalize this welfare-robustness trade-off via the \emph{Exploitability Welfare Frontier} (EWF) in the next section.

\subsection{Exploitability Welfare Frontier (EWF)}

We plot, in \Cref{fig:pareto-frontier}, the maximum achievable social welfare as a function of the exploitability $\epsilon$. We refer to this curve as the \emph{Exploitability Welfare Frontier (EWF)}, defined below.

\begin{definition}[Exploitability Welfare Frontier (EWF)]
\begin{enumerate}[left=0mm,nosep,leftmargin=*, align=left, labelsep=0pt, labelwidth=0pt,label={},ref={$\ewf^{\rm CCE}$}] 
\item \label{def:exploitability-welfare-frontier} The exploitability welfare frontier for CCE, denoted $\ewf^{\rm CCE}(\epsilon)\colon [0,1]\to[0,N]$, is the maximum social welfare attainable among all $\epsilon$-approximate CCEs. Formally,
\begin{align}
    \ewf^{\rm CCE}(\epsilon) \coloneqq \max_{\pi\in \epsilon\text{-CCE}} \sum_{i=1}^N \EE_{\ba\sim\pi}\sbr{\cU_i(\ba)}.
\end{align}
where
\begin{align}
    \epsilon\text{-CCE}\coloneqq \cbr{\pi\in\Delta^{\cA}\colon \max_{i\in [N]} \max_{\hat a_i\in\cA_i} \EE_{\ba\sim \pi}\!\sbr{\cU_i(\hat a_i, \ba_{-i}) - \cU_i(\ba)}\leq \epsilon}.
\end{align}
\end{enumerate}
\end{definition}
Analogously, one can define $\ewf^{\rm NE}$ and $\ewf^{\rm CE}$ for Nash and correlated equilibria, respectively.

\begin{wrapfigure}[13]{r}{0.5\textwidth}
  \centering
  \includegraphics[width=0.35\textwidth]{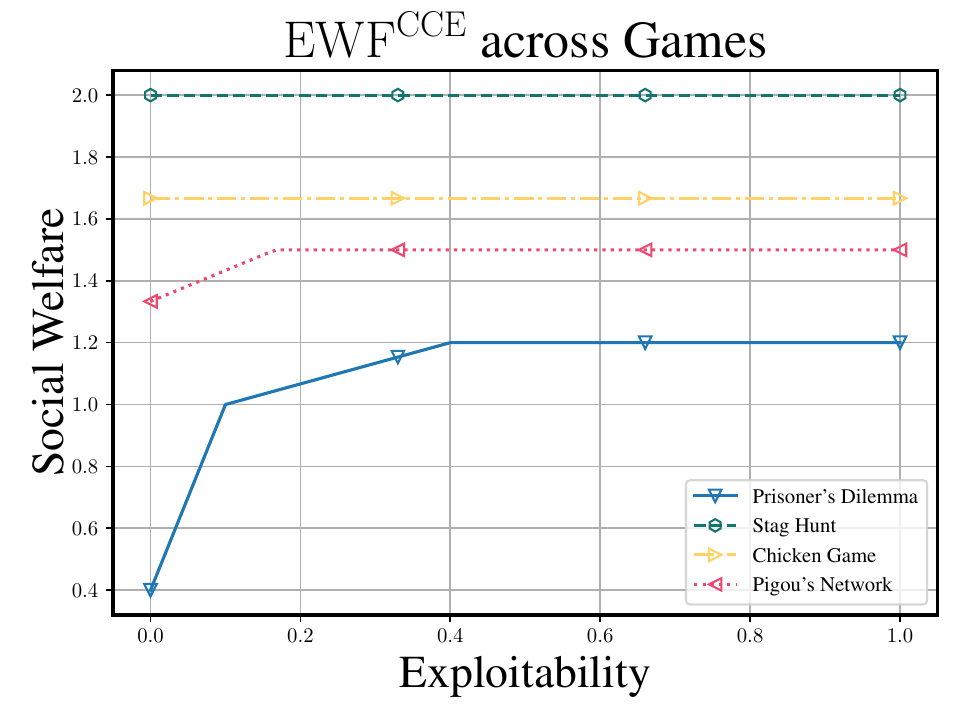}
  \caption{$\ewf^{\rm CCE}$ over four different games.}
  \label{fig:pareto-frontier}
\end{wrapfigure}
Prior work has investigated related welfare–approximation trade-offs for NE (\emph{e.g.}, in congestion games \citep{christodoulou2011performance-congestion-tradeoff} and two-player normal-form games \citep{czumaj2015approximate-bilinear}). However, computational barriers limit what is achievable. \citet{czumaj2015approximate-bilinear}'s algorithm relies on an efficient oracle for computing approximate NE, yet computing NE is {\sf PPAD}-hard \citep{chen2006settling-ppad,daskalakis2009complexity-PPAD}. Moreover, \citet{deligkas2016inapproximability-bilinear-hard} shows that finding the welfare-maximizing approximate NE can require superpolynomial time even in two-player normal-form games.
Hence, to ensure computational tractability, we leave these extensions out of scope and focus on $\ewf^{\rm CCE}$ throughout.

We establish basic structural properties of \Cref{def:exploitability-welfare-frontier} in the following lemmas.
\begin{lemma}
    \label{lemma:pareto-frontier-concavity}
    In any game, \Cref{def:exploitability-welfare-frontier} is non-decreasing, concave, and piecewise linear as a function of $\epsilon$.
\end{lemma}

\begin{lemma}
\label{lemma:slope-unbounded}
    The slope of the \Cref{def:exploitability-welfare-frontier} can be as small as zero and is unbounded above.
\end{lemma}

The proof is deferred to \Cref{section:proof-tradeoff}. By \Cref{lemma:pareto-frontier-concavity}, the \Cref{def:exploitability-welfare-frontier} is concave. Hence, its slope is non-increasing in $\epsilon$ wherever it is differentiable. The same concavity argument also yields the following lower bound on the price of stability \citep{nisan2007algorithmic-game-theory} of $\epsilon$-CCE.

\begin{corollary}
\label{corollary:price-of-stability}
For any finite game with positive optimal social welfare and any $\epsilon\in\sbr{0,1}$, the price of stability of $\epsilon$-CCE satisfies
\begin{align*}
    \frac{\max_{\pi\in \epsilon\text{-CCE}} \sum_{i=1}^N \EE_{\ba\sim\pi}\sbr{\cU_i(\ba)}}{\max_{\pi\in \Delta^{\cA}} \sum_{i=1}^N \EE_{\ba\sim\pi}\sbr{\cU_i(\ba)}}\geq \epsilon.
\end{align*}
\end{corollary}
\Cref{corollary:price-of-stability} follows directly from the concavity of the EWF function.

\subsection{Finding a Pareto Optimal Strategy}
\label{section:binary-search}

A natural next question is: given a tolerance $\epsilon\geq 0$ on exploitability, how do we compute a strategy that maximizes social welfare among all $\epsilon$-approximate CCEs? We show that this can be computed efficiently using a variant of our MASE framework. Specifically, we solve the following weighted objective:
\begin{equation}
    \argmin_{\pi\in \Delta^{\cA}}\max_{S\in \cS} \max_{\hat\ba_S\in\cA_S} \frac{w_S}{|S|}\sum_{i\in S} \EE_{\ba\sim \pi}\sbr{\cU_i\rbr{\hat\ba_S, \ba_{-S}} - \cU_i\rbr{\ba}}, \label{eq:weighted-mase}
\end{equation}
where $\bw\in\RR^{\cS}$ is a vector of non-negative weights. We have the following lemma.
\begin{lemma}
\label{lemma:equivalence-tradeoff}
    For any $\epsilon > 0$, computing the CCE with exploitability no more than $\epsilon$ that maximizes social welfare is equivalent to \Cref{eq:weighted-mase} by setting $\cS=\cbr{\cbr{i}}_{i\in [N]}\cup \cbr{[N]}$ and using the weights:
    \begin{align*}
        w_S=\begin{cases}
            w & \text{if } |S|=1 \\
            1-w & \text{if } S=[N]
        \end{cases}
    \end{align*}
    for some $w \in [0, 1)$. Conversely, solving \Cref{eq:weighted-mase} with these parameters corresponds to finding a point on the Pareto frontier of social welfare and exploitability.
\end{lemma}
The proof is postponed to \Cref{section:proof-tradeoff}. In \Cref{fig:tradeoff-1}, we plot the exploitability and social welfare of the strategy induced by varying $w$. For the Stag Hunt, because one of the Nash equilibria already attains maximal social welfare, the social welfare remains fixed at its optimal value for all $w \in [0, 1)$.

\begin{figure}[h]
    \centering
    \includegraphics[width=0.95\linewidth]{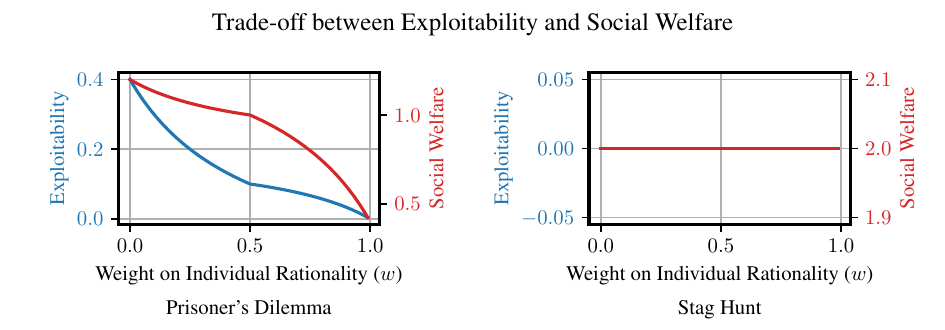}
    \caption{The trade-off between exploitability and social welfare in Prisoner's Dilemma and the Stag Hunt.}
    \label{fig:tradeoff-1}
\end{figure}

With \Cref{lemma:equivalence-tradeoff}, we can compute $\ewf^{\rm CCE}$ by solving \Cref{eq:weighted-mase} for different values of $w$. This also yields an efficient algorithm (fixed-parameter tractable in the treewidth of the \Cref{def:utility-dependence-graph}, as in \Cref{section:computation-MASE}) for computing the socially optimal approximate CCE. 

Fix a target tolerance $\epsilon>0$, We can recover an optimal $\epsilon$-approximate CCE as follows.
By \Cref{lemma:equivalence-tradeoff}, each choice of $w\in[0,1)$ induces a point on the \Cref{def:exploitability-welfare-frontier}, with the resulting solution attaining some exploitability level $\epsilon'$. Moreover, $\epsilon'$ decreases monotonically as $w$ increases. This monotone correspondence suggests a simple approach: perform a binary search over $w$ to find the optimal strategy whose exploitability does not exceed the target $\epsilon$. The details are given in \Cref{alg:binary-search}.

\begin{algorithm}[h]
\caption{Binary Search}
\label{alg:binary-search}
\SetAlgoLined

\KwInput{The exploitability tolerance $\epsilon$ and the floating-point tolerance ${\rm EPS}$.}
\KwOutput{$\ewf^{\rm CCE}(\epsilon)$.}

$l\gets 0, r\gets 1$.

\While{$r-l>{\rm EPS}$}{
    $w\gets \frac{l+r}{2}$

    Solve \Cref{eq:weighted-mase} associated with $w$.

    Compute the exploitability $\epsilon'$ and the social welfare $sw$ of the resulting strategy.

    \eIf{$\epsilon' \leq \epsilon$}{
        $sw^{\rm final}\gets sw$
        
        $r \gets w$
    }{
        $l\gets w$
    }
}

\KwReturn{$sw^{\rm final}$.}

\end{algorithm}

Despite the tractability of our approach when \Cref{def:utility-dependence-graph} has a bounded treewidth, computing the \Cref{def:exploitability-welfare-frontier} is intractable in general succinct games: even evaluating the value at $\epsilon=0$ (\emph{i.e.}, maximum welfare over exact CCEs) is {\sf NP}-hard in a number of standard succinct game classes.
\begin{lemma}
\label{lemma:ewf-hardness}
    Computing $\ewf^{\rm CCE}(0)$ is {\sf NP}-hard in each of the following succinct games:
    \begin{enumerate}
        \item Bounded-degree, bipartite graphical games.
        \item Polymatrix games.
        \item Hypergraphical games.
        \item Congestion games.
        \item Local effect games.
        \item Facility location games.
        \item Network design games.
        \item Scheduling games.
    \end{enumerate}
\end{lemma}
The proof is postponed to \Cref{section:proof-tradeoff}. Still, it may be possible to evaluate $\ewf^{\rm CCE}(\epsilon)$ efficiently when $\epsilon$ is sufficiently large. For example, \citet{tsaknakis2007optimization-large-epsilon} gave a polynomial-time algorithm for computing an $\epsilon$-approximate Nash equilibrium (with exploitability at most $\epsilon$) for $\epsilon \geq 0.3393$, whereas computing approximate Nash equilibria in general is {\sf PPAD}-hard \citep{chen2006settling-ppad,daskalakis2009complexity-PPAD}. This suggests that, for large enough $\epsilon$, computing $\ewf^{\rm CCE}(\epsilon)$ may be computationally tractable.

\section{Conclusion}

In this work, we introduced the Minimum Average-Strong Equilibrium (MASE), a tractable solution concept that accounts for multilateral deviations by minimizing each coalition’s average incentive to deviate. We established that computing an approximate MASE is {\sf NP}-hard even with singleton coalitions and proved a lower bound showing unavoidable exponential dependence on the treewidth of the \Cref{def:utility-dependence-graph}. Complementing these hardness results, we design an algorithm that combines a correlator-deviator meta-game with FTPL updates and dynamic programming over a tree decomposition. Its exponential dependence on the treewidth matches the lower bound. We also showed how this framework can be used to compute the exploitability welfare frontier (EWF). Empirically, MASE is substantially more robust to coalition deviations than standard baselines, while improving social welfare in canonical games and often incurring little to no increase in unilateral exploitability. Finally, we showed that minimizing the minimum deviation gain within each coalition, \emph{i.e.}, computing a minimum strong CCE, is also {\sf NP}-hard, even when the joint action space is polynomial in the number of players. This further motivates our focus on minimizing average coalitional gain.

\bibliographystyle{plainnat}
\bibliography{main}

\newpage
\appendix

\section{Linear Programming for Solving \Cref{eq:MASE}}
\label{section:lp}

\Cref{eq:MASE} can be solved by the following linear programming
\begin{align*}
    &\min_{\pi,w} w\\
    &w\geq \frac{1}{|S|}\sum_{i\in S} \EE_{\ba\sim \pi}\sbr{\cU_i\rbr{\hat\ba_S, \ba_{-S}} - \cU_i\rbr{\ba} }~~~\forall S\in\cS,\hat\ba_S\in\cA_S\\
    &\pi(\ba)\geq 0~~~\forall \ba\in\cA\\
    &\sum_{\ba\in\cA} \pi(\ba)=1.
\end{align*}
Since
\begin{align*}
    &\EE_{\ba\sim\pi}\sbr{\cU_i\rbr{\hat\ba_S, \ba_{-S}}}=\sum_{\ba\in\cA} \pi(\ba)\cU_i(\hat\ba_S,\ba_{-S})\\
    &\EE_{\ba\sim\pi}\sbr{\cU_i\rbr{\ba}}=\sum_{\ba\in\cA} \pi(\ba)\cU_i(\ba)
\end{align*}
are linear in $\pi$, the linear programming above is valid. Note that the linear program contains exponentially many variables ($\pi \in \Delta^{\cA}$), so its complexity is necessarily exponential in $N$.

\section{Omitted Proofs in \Cref{section:relation}}
\label{section:relation-proof}

\restate{lemma:non-existence-coalition-proof}

\begin{proof}
    We construct a counterexample using a variant of the game described by \citet{bernheim1987coalition-proof}. Consider three players dividing five coins. Each player $i$ chooses an action from the set of all possible strict rankings over the players, $\cA_i = \cbr{ (P_1 \succ P_2 \succ P_3), \dots, (P_3 \succ P_2 \succ P_1) }$. If two or more players vote for the same ranking, the payoffs are distributed according to that ranking: the first-ranked player receives 3 coins, the second-ranked player receives 2 coins, and the third-ranked player receives 0. If no majority agreement is reached (\emph{i.e.}, all three players vote for distinct rankings), all players receive 0.

    Let $\pi^*$ be any CCE of this game. By definition, any coalition-proof CCE must also be a CCE. We prove the non-existence of a coalition-proof CCE by demonstrating that any valid CCE $\pi^*$ admits a self-enforcing deviation by a coalition.

    Assume without loss of generality that the expected utilities of players 1, 2, and 3 are non-increasing:
    \begin{align*}
        \EE_{\ba\sim\pi^*}\sbr{\cU_1(\ba)}\geq \EE_{\ba\sim\pi^*}\sbr{\cU_2(\ba)}\geq \EE_{\ba\sim\pi^*}\sbr{\cU_3(\ba)}.
    \end{align*}

    Consider the coalition $S=\cbr{2, 3}$ and a deviation strategy where both members deterministically vote for the ranking $P_2 \succ P_3 \succ P_1$. Under this deviation, players 2 and 3 form a majority regardless of player 1's action, guaranteeing themselves payoffs of 3 and 2, respectively. This deviation is both valid and self-enforcing:
    \begin{enumerate}
        \item \textbf{Improvement:} The total sum of utilities in any outcome is at most 5. Given the non-increasing order of utilities, it must hold that player 2's original expected utility is strictly less than 3 ($ \EE_{\ba\sim\pi^*}\sbr{\cU_2(\ba)}<3$) and player 3's is strictly less than 2 ($ \EE_{\ba\sim\pi^*}\sbr{\cU_3(\ba)}<2$). Thus, the deviation strictly improves the payoff for both coalition members.
        \item \textbf{Self-Enforcement:} The deviation constitutes a coalition-proof CCE in the reduced game restricted to coalition $S$. Player 2 has no incentive to deviate further, as player 2 already achieves the maximum possible utility of 3 coins. Player 3 also has no incentive to deviate. Any further deviation by Player 3 breaks the agreement with Player 2; thus, Player 3 can only accrue positive utility by matching Player 1's strategy, $\pi^*_1$. However, if coordinating with $\pi^*_1$ were to yield a utility strictly greater than 2, Player 3 would have had a profitable unilateral deviation from the original equilibrium $\pi^*$. This contradicts the assumption that $\pi^*$ is a CCE.
    \end{enumerate}
    
    Since there exists a valid, self-enforcing deviation for the coalition $\cbr{2, 3}$, $\pi^*$ cannot be a coalition-proof CCE. As this holds for any CCE $\pi^*$, no coalition-proof CCE exists in this game. \qedhere
\end{proof}

\section{Omitted Proofs in \Cref{section:hardness}}
\label{section:hardness-proof}

This section presents the omitted proofs in \Cref{section:hardness}.

\subsection{Proof of \Cref{theorem:complexity-strong-equilibrium}}
\label{section:proof-mase-hardness}

\restate{theorem:complexity-strong-equilibrium}

\begin{proof}

We will introduce the allocation problem ({\sf NP}-hard) and show that it can be reduced to computing the correlated strong equilibrium.

\begin{definition}[Allocation Problem]
    There are $n$ agents and $m$ goods. An assignment $X\colon [m]\to [n]$ is a mapping from each good to an agent. Agent $i$'s utility is $u_i(X)$ for an assignment $X$. A stochastic allocation $\bp\in\Delta^{[m]^{[n]}}$ is a distribution over all possible assignments. The egalitarian social welfare (ESW) maximization is defined as
    \begin{align}
        \max_{\bp\in\Delta^{[m]^{[n]}}} \min_{i\in [n]} \sum_{X\in [m]^{[n]}} p(X)u_i(X).\label{eq:def-allocation}
    \end{align}
\end{definition}

    For any allocation problem, we can create a game with $N=n+m$ players. The action set of player $i\leq n$ is $\cA_i=\cbr{0,1}$, while the action set of player $j>n$ is $\cA_j=[n]$. For any joint action $\ba\in\cA$, $\cU_i(\ba)=u_i(\ba_{-[n]})$ for $i\leq n$ when $a_1=a_2=\dots=a_n$, otherwise $\cU_i(\ba)=-u_i(\ba_{-[n]})$. For $j>n$, $\cU_j(\ba)=0$. Moreover, let $\cS=\cbr{\cbr{1}, \cbr{2},\dots,\cbr{n}}$. We further define
    \begin{align*}
        a\coloneqq& \sum_{\substack{\ba\in\cA\colon\\a_2=a_3=\dots=a_n=0}}\pi(\ba)u_i\rbr{\ba_{-[n]}}\\
        b\coloneqq& \sum_{\substack{\ba\in\cA\colon\\a_2=a_3=\dots=a_n=1}}\pi(\ba)u_i\rbr{\ba_{-[n]}}\\
        c\coloneqq& -\sum_{\substack{\ba\in\cA\colon\\\exists i,j\in [n]\setminus\cbr{1},a_i\not=a_j}}\pi(\ba)u_i\rbr{\ba_{-[n]}}.
    \end{align*}
    Then, the gap of player $1$ is lower bounded by
    \begin{align*}
        \max_{\hat a_1\in\cbr{0,1}}\EE_{\ba\sim\pi}\sbr{\cU_i\rbr{\hat a_1, \ba_{-1}} - \cU_i\rbr{\ba}}\geq& \max\rbr{a,b}+c - (a+b+c)=-\min\rbr{a,b}.
    \end{align*}
    The equation holds when $a_1$ is always equal to $a_2$ when $a_2=a_3=\dots=a_n$.
    Therefore, the optimal strategy $\pi$ should satisfy $\pi_{[n]}((0,\dots 0))=\pi_{[n]}((1,\dots 1))=\frac{1}{2}$. 

    Then, \Cref{eq:MASE} is equivalent to 
    \begin{align*}
        \min_{\pi_{[m]}\in \Delta^{[m]^{[n]}}} \max_{i\in [n]} \sum_{\ba\in \cA} - \pi(\ba) u_i(\ba_{-[n]})=\min_{\pi_{[m]}\in \Delta^{[m]^{[n]}}} \max_{i\in [n]} \sum_{\ba_{-[n]}\in [m]^{[n]}} - \pi_{-[n]}(\ba_{-[n]}) u_i(\ba_{-[n]}),
    \end{align*}
    which is equivalent to \Cref{eq:def-allocation}. Finally, according to \citet[Corollary 1]{kawase2020max-hardness-allocation}, it is {\sf NP}-hard to approximate \Cref{eq:def-allocation} up to $1-\frac{1}{e}$.

    The hard instance constructed in \citet{kawase2020max-hardness-allocation} satisfies that $\max_i\max_{X\in [m]^{[n]}} u_i(X)=Poly(n,m)$ and the solution to \Cref{eq:def-allocation} is $1$. Therefore, $\text{Poly}(N,\frac{1}{\epsilon})$ algorithm does not exist unless {\sf P=NP} for solving an $\epsilon$-MASE.
\end{proof}

\subsection{Proof of \Cref{theorem:tree-width-important}}
\label{section:tree-width-proof}

\restate{theorem:tree-width-important}

\begin{proof}
    According to \citet{lokshtanov2011known-treewidth-complexity}, under SETH, $q$-coloring cannot be solved in $\cO((q-\zeta)^{\rm tw}\cdot \text{Poly}(|I|))$ for arbitrary graph $G$, when a tree decomposition of width $tw$ is given.\footnote{As summarized in \citet{esmer2024fundamental-q-coloring}, the proof of the $q$-coloring complexity implicitly implies that the complexity is lower bounded by $\cO(q^{tw}\cdot \text{Poly}(|I|))$, even though a tree decomposition of width $tw$ is given. In other words, aside from computing a tree decomposition, the $q$-coloring itself has an intrinsic computational barrier.} In the sequel, we construct a game such that computing $\frac{1}{N}$-approximate \Cref{eq:MASE} is equivalent to determining the $q$-coloring.

    For any $q$-coloring problem on $G=(V, E)$, we will construct a game with $N=|V|+|E|$ players. For each player $i\leq |V|$, the action set $\cA_i=\cbr{1,2,\dots,q}$ and the utility function $\cU_i\equiv 0$ is a constant function equal to zero. For each player $j>|V|$, the action set is $\cbr{1}$ and the utility function is $\cU_j(\ba)=\ind\rbr{a_{e_{j-|V|,1}}\not=a_{e_{j-|V|,2}}}$, where $(e_{j-|V|, 1}, e_{j-|V|, 2})$ is the $(j-|V|)^{th}$ edge in $E$ and $\ind$ is the indicator function (equals one when the argument is true and otherwise zero). In this game, $\cS=\cbr{\cbr{i_1,i_2,j+|V|}\given e_j=(i_1,i_2)}$.

    Firstly, for any proper coloring $\bc\in [q]^{|V|}$, the associated pure strategy is $\pi^{\bc}$, where $\pi_i^{\bc}(a_i)=1$ if and only if $a_i=c_i$ and $0$ otherwise. It satisfies \Cref{eq:MASE}. Because for any $\cS\ni S=(i_1,i_2,j+|V|)$, the maximum of $\frac{1}{|S|}\sum_{i\in S} \cU_i(\pi)$ is $\frac{1}{3}$, which is attained when the colors of node $i_1,i_2$ are different. Therefore, $\pi^{\bc}$ obtains the maximum for every coalition $S\in\cS$, which implies the satisfaction of \Cref{eq:MASE}.

    Secondly, for any joint strategy $\pi\in \Delta^{\cA}$ satisfying \Cref{eq:MASE}, we have 
    \begin{align*}
        \min_{S\in\cS} \frac{\sum_{i\in S} \cU_i(\pi)}{|S|} \leq \frac{1}{|\cS|} \sum_{S\in\cS} \frac{\sum_{i\in S} \cU_i(\pi)}{|S|} =& \sum_{\ba\in\cA} \pi(\ba) \frac{1}{|\cS|} \sum_{S\in\cS} \frac{\sum_{i\in S} \cU_i(\ba)}{|S|}\\
        \leq& \frac{1}{3}-\frac{1}{3|E|}\sum_{\ba\in\cA} \pi(\ba) \ind\rbr{\ba_{[|V|]}\text{ is not a proper coloring}}.
    \end{align*}
    Because there must exist at least one edge with both of its nodes in the same color for any improper coloring. On the other hand, let $\hat\pi=\pi^{\bc}$ for some proper coloring $\bc$. Then, for any $S\in\cS$, we have
    \begin{align*}
        \frac{\sum_{i\in S} \cU_i(\hat\pi)}{|S|} = \frac{1}{3}.
    \end{align*}
    Therefore, the approximation error of \Cref{eq:MASE} is at least
    \begin{align*}
        \frac{1}{3|E|}\sum_{\ba\in\cA} \pi(\ba) \ind\rbr{\ba_{[|V|]}\text{ is not a proper coloring}}
    \end{align*}
    for any joint strategy $\pi$. When we get $\frac{1}{9N^2}$ approximation, since $|E|\leq N^2$, we have
    \begin{align*}
        \sum_{\ba\in\cA} \pi(\ba) \ind\rbr{\ba_{[|V|]}\text{ is not a proper coloring}}\leq \frac{1}{3}.
    \end{align*}
    Therefore, when sampling $\ba\sim \pi$, we will get a proper coloring with probability at least $\frac{2}{3}$, which is in complexity class {\sf RP}. As a result, when {\sf P=RP}, the time complexity of computing \Cref{eq:MASE} is at least $O^*(A^{tw(G)})$, where $G$ is the trust graph and $A$ is the size of the maximal action set.
\end{proof}

\subsection{Proof of \Cref{theorem:strong-CCE-hardness}}
\label{section:strong-CCE-hardness-proof}

\restate{theorem:strong-CCE-hardness}

\begin{proof}
    The proof is inspired by \citet[Theorem 1]{borgs2010myth-folk-theorem}, who showed that approximating
    \begin{align*}
        \min_{\pi_1\in\Delta^{\cA_1}}\min_{\pi_2\in\Delta^{\cA_2}}\max_{\pi_3\in\Delta^{\cA_3}} \EE_{\ba\sim (\pi_1,\pi_2,\pi_3)}\sbr{\cU_3(\ba)}
    \end{align*} 
    is {\sf NP}-hard via a reduction from 3-colorability.

    We similarly reduce 3-colorability to the problem of finding an approximate \Cref{eq:strong-CCE}. Given any graph $G=(V,E)$ with $|V|\geq 4$, we construct a game with $N=3|V|+2$ players and $\cS=\cbr{S}$, where $S=\cbr{1,2,\dots, 3|V|, 3|V|+2}$. Thus, there is only a single coalition of interest. The objective of strong CCE is therefore
    \begin{align}
        &\min_{\pi\in \Delta^{\cA}} ~\max_{\hat\pi_S\in\Delta^{\cA_S}} ~\min_{i\in S} \EE_{\ba\sim \pi, \hat\ba_S\sim \hat\pi_S}\sbr{\cU_i\rbr{\hat\ba_S, \ba_{-S}} - \cU_i\rbr{\ba} }\notag\\
        \overset{(i)}{\Leftrightarrow}& \min_{\pi\in \Delta^{\cA}} ~\max_{\hat\pi_S\in\Delta^{\cA_S}} ~\min_{\balpha\in \Delta^{\cS}} \sum_{i\in S} \alpha_i~ \EE_{\ba\sim \pi, \hat\ba_S\sim \hat\pi_S}\sbr{\cU_i\rbr{\hat\ba_S, \ba_{-S}} - \cU_i\rbr{\ba} }\notag\\
        \overset{(ii)}{\Leftrightarrow}& \min_{\pi\in \Delta^{\cA}} ~\min_{\balpha\in \Delta^{\cS}} ~\max_{\hat\pi_S\in\Delta^{\cA_S}} \sum_{i\in S} \alpha_i~ \EE_{\ba\sim \pi, \hat\ba_S\sim \hat\pi_S}\sbr{\cU_i\rbr{\hat\ba_S, \ba_{-S}} - \cU_i\rbr{\ba} }. \label{eq:strong-CCE-proof-1}
    \end{align}
    $(i)$ holds because
    \begin{align*}
        \sum_{i\in S} \alpha_i~ \EE_{\ba\sim \pi, \hat\ba_S\sim \hat\pi_S}\sbr{\cU_i\rbr{\hat\ba_S, \ba_{-S}} - \cU_i\rbr{\ba} }
    \end{align*}
    is linear in $\balpha$, so the minimum is attained at an extreme point. In other words, $\balpha$ places all of its mass on a single player in $S$. $(ii)$ follows from the minimax theorem.

    Intuitively, we construct the game so that $\pi$, $\balpha$, and $\hat\pi_S$ correspond to the strategies of three different agents. Agent 1 and agent 2 each choose a vertex in $V$ together with a color in $[3]$. Agent 3 then chooses either agent 1 or agent 2 and tries to predict the vertex selected by that agent. We will show that, under the utility functions defined below, the value of \Cref{eq:strong-CCE-proof-1} differs depending on whether $G$ is 3-colorable. Consequently, if we can approximate this value with sufficient accuracy, then we can decide whether $G$ is 3-colorable.
    
   The action sets of the game are defined by
    \begin{align*}
        \cA_i=\begin{cases}
            3|V|&i=3|V|+1\\
            2|V|&i=3|V|+2\\
            1&i\leq 3|V|,
        \end{cases}
    \end{align*}
    for all $i\in [N]$. Hence, every joint action can be written as $((v_1, c_1), (v_3, i_3))$, where $v_1,v_3\in V$, $c_1\in[3]$, and $i_3\in[2]$. Thus, agent 1 proposes a vertex $v_1$ together with a color $c_1$, while agent 3 predicts that player $i_3$ proposed vertex $v_3$. In addition, for each player $i\leq 3|V|$, we encode its index as $(v_2,c_2)$ with $v_2\in V$ and $c_2\in[3]$. This represents the proposal of agent 2. For any joint action $((v_1,c_1),(v_3,i_3))$, the utility functions are defined by
    \begin{align*}
        &\cU_{3|V|+1}\rbr{((v_1, c_1), (v_3, i_3))} = \cU_{3|V|+2}\rbr{((v_1, c_1), (v_3, i_3))} = 1\\
        &\cU_{(v_2, c_2)}\rbr{((v_1, c_1), (v_3, i_3))}=\begin{cases}
            1&v_1=v_2,~~c_1\neq c_2\\
            1&(v_1, v_2)\in E,~~c_1= c_2\\
            1&v_{i_3} = v_3\\
            0&\text{Otherwise.}
        \end{cases}
    \end{align*}
    In words, a player receives a reward $1$ whenever the proposals of agents 1 and 2 are inconsistent, or whenever agent 3 makes a correct prediction. Hence, in the corresponding \Cref{def:utility-dependence-graph}, the edge set is
    \begin{align*}
        \cE=\cbr{(i, 3|V|+1)}_{i=1}^{3|V|} \cup \cbr{(i, 3|V|+2)}_{i=1}^{3|V|} \cup \cbr{(3|V|+1, 3|V|+2)}.
    \end{align*}
    In other words, every one of the first $3|V|$ vertices is adjacent to both $3|V|+1$ and $3|V|+2$, and these latter two vertices are also adjacent. Hence, the resulting graph has treewidth $2$.

   Now fix any joint action $((v_1,c_1),(v_3,i_3))$. We can factor 
   \begin{align*}
       \pi(((v_1,c_1),(v_3,i_3))) = \pi_{-S}((v_1,c_1))\,\pi_S((v_3,i_3)\given (v_1,c_1)),
   \end{align*}
   and only the second factor $\pi_S((v_3,i_3)\given (v_1,c_1))$ affects the term $\cU_i(\ba)$ in \Cref{eq:strong-CCE-proof-1}. Since $\pi$ is maximizing $\cU_i(\ba)$, we may always choose
    \begin{align*}
        \pi_S((v_3, i_3)\given (v_1, c_1)) = \ind\rbr{v_1=v_3\text{ and } i_3=1}.
    \end{align*}
    Therefore, for every $i\leq 3|V|$, we have $\cU_i(\ba)=1$ in \Cref{eq:strong-CCE-proof-1}. For $i>3|V|$, we also have $\cU_i(\ba)=1$ by definition of the utility functions. Thus, $\cU_i(\ba)$ is constant in \Cref{eq:strong-CCE-proof-1} and can be removed from the optimization objective. We may therefore rewrite the objective as
    \begin{align}
        &\min_{\pi_{-S}\in \Delta^{\cA_{-S}}} ~\min_{\balpha\in \Delta^{\cS}} ~\max_{\hat\pi_S\in\Delta^{\cA_S}} \sum_{i\in S} \alpha_i~ \EE_{\ba_{-S}\sim \pi_{-S}, \hat\ba_S\sim \hat\pi_S}\sbr{\cU_i\rbr{\hat\ba_S, \ba_{-S}}}\notag\\
        \overset{(i)}{\Leftrightarrow}&\min_{\pi_{-S}\in \Delta^{\cA_{-S}}} ~\min_{\balpha\in \Delta^{[3|V|]}} ~\max_{\hat\pi_S\in\Delta^{\cA_S}} \sum_{i=1}^{3|V|} \alpha_i~ \EE_{\ba_{-S}\sim \pi_{-S}, \hat\ba_S\sim \hat\pi_S}\sbr{\cU_i\rbr{\hat\ba_S, \ba_{-S}}}. \label{eq:strong-CCE-proof-2}
    \end{align}
    $(i)$ holds because $\cU_i(\ba)=1$ for every $\ba\in\cA$ whenever $i>3|V|$. Hence, we may assume without loss of generality that $\alpha(3|V|+2)=0$, since $\cU_j(\ba)\leq \cU_i(\ba)=1$ for any $\ba\in\cA$, $i>3|V|$, and $j\leq 3|V|$.

    We now analyze \Cref{eq:strong-CCE-proof-2}. First, it is always at least $\frac{1}{|V|}$, because $\hat\pi_S$ (agent 3) can choose $i_3=1$ and select a vertex $v_3\in V$ uniformly at random. Then,
    \begin{align*}
        \Pr_{(v_1,c_1)\sim \pi_{-S}, (v_3, i_3)\sim \hat \pi_S}(v_1=v_3)=\sum_{v\in V} \rbr{\sum_{c_1=1}^3 \pi_{-S}((v, c_1))} \hat\pi_S((v, 1))=\frac{1}{|V|} \sum_{v\in V} \sum_{c_1=1}^3 \pi_{-S}((v, c_1)) = \frac{1}{|V|}.
    \end{align*}
    If $G$ is 3-colorable, then \Cref{eq:strong-CCE-proof-2} is exactly $\frac{1}{|V|}$: agent 1, represented by $\pi_{-S}$, and agent 2, represented by $\balpha$, both choose a uniformly random vertex and then assign to that vertex the color given by the same valid 3-coloring. Under this choice, the proposals of agents 1 and 2 are always consistent. Therefore, the only way for a utility term to equal $1$ is for agent 3 to predict the chosen vertex correctly, which occurs with probability exactly $\frac{1}{|V|}$.

    In contrast, suppose that $G$ is not 3-colorable. If either agent 1 or agent 2 assigns probability at least $\frac{1}{|V|}+\frac{1}{3|V|^2}$ to some vertex, then \Cref{eq:strong-CCE-proof-2} is at least $\frac{1}{|V|}+\frac{1}{3|V|^2}$ by letting agent 3 guessing that vertex deterministically.

    It remains to consider the case in which both agent 1 and agent 2 assign a probability at most $\frac{1}{|V|}+\frac{1}{3|V|^2}$ to every vertex. Then, for any $v_1\in V$,
    \begin{align*}
        \sum_{c_1=1}^3 \pi_S((v_1, c_1))\geq 1-(|V|-1)\rbr{\frac{1}{|V|} + \frac{1}{3|V|^2}}\geq \frac{2}{3|V|},
    \end{align*}
    and similarly,
    \begin{align*}
        \sum_{c_2=1}^3 \alpha((v_2,c_2))\geq \frac{2}{3|V|}.
    \end{align*}
    Using the shorthand
    \begin{align*}
        \pi_{-S}(v_1)=\sum_{c_1=1}^3 \pi_{-S}((v_1, c_1))
    \end{align*}
    for the marginal distribution on vertices, and similarly for $\balpha$, we obtain
    \begin{align*}
        &\Pr_{(v_1,c_1)\sim\pi_{-S}, (v_2,c_2)\sim \balpha}\rbr{v_1=v_2,~~c_1\neq c_2\text{ or } (v_1, v_2)\in E,~~c_1= c_2}\\
        =&\sum_{v_1, v_2\in V} \pi_1(v_1) \alpha(v_2) \Pr_{c_1\sim\pi_{-S}(\cdot\given v_1), c_2\sim\alpha(\cdot \given v_2)}\rbr{v_1=v_2,~~c_1\neq c_2\text{ or } (v_1, v_2)\in E,~~c_1= c_2}\\
        \geq& \frac{4}{9|V|^2}\sum_{v_1, v_2\in V} \Pr_{c_1\sim\pi_{-S}(\cdot\given v_1), c_2\sim\alpha(\cdot \given v_2)}\rbr{v_1=v_2,~~c_1\neq c_2\text{ or } (v_1, v_2)\in E,~~c_1= c_2\given v_1, v_2}\\
        \overset{(i)}{\geq}& \frac{4}{9|V|^2}.
    \end{align*}
    $(i)$ holds because
    \begin{align*}
        \sum_{v_1, v_2\in V} \Pr_{c_1\sim\pi_{-S}(\cdot\given v_1), c_2\sim\alpha(\cdot \given v_2)}\rbr{v_1=v_2,~~c_1\neq c_2\text{ or } (v_1, v_2)\in E,~~c_1= c_2\given v_1, v_2}
    \end{align*}
    can be interpreted as the expected number of violations over all pairs of coloring plans. Since $G$ is not 3-colorable, any two coloring plans must contain at least one violation: either they assign different colors to the same vertex, or they assign the same color to two adjacent vertices. Hence, every pair of coloring plans has at least one inconsistency, and therefore, this expectation is at least one.

    Finally, let agent 3 guess agent 1's proposed vertex uniformly at random. Then \Cref{eq:strong-CCE-proof-2} is at least
    \begin{align*}
        \frac{4}{9|V|^2} + \rbr{1-\frac{4}{9|V|^2}.}\frac{1}{|V|} = \frac{1}{|V|} + \frac{4}{9|V|^2}\rbr{1-\frac{1}{|V|}}\overset{(i)}{\geq} \frac{1}{|V|} + \frac{1}{3|V|^2}.
    \end{align*}
    $(i)$ holds whenever $|V|\ge 4$. Therefore, if we can approximate \Cref{eq:strong-CCE} within $\frac{1}{N^2}<\frac{1}{3|V|^2}$, then we can decide whether $G$ is 3-colorable. \qedhere
\end{proof}

\section{Proof of \Cref{theorem:existence}}
\label{section:proof-existence}

\restate{theorem:existence}
\begin{proof}

Let $D\coloneqq \sum_{S\in\cS}\sum_{i\in S}\abr{\cA_{S\cap \cN(i)}}$. For any joint strategy $\pi\in\Delta^{\cA}$, consider the vector $\bv^{\pi}\in \RR^D$, where
\begin{align*}
    v^{\pi}(S, i, \hat \ba_{S\cap \cN(i)})=\EE_{\ba\sim\pi}\sbr{\cU_i(\hat \ba_{S\cap \cN(i)}, \ba_{-\rbr{S\cap \cN(i)}}) - \cU_i(\ba)},
\end{align*}
for any $S\in\cS$ and $\hat \ba_{S\cap \cN(i)}\in\cA_{S\cap \cN(i)}$. By definition of $\bv^{\pi}$, it is linear in $\pi$. Therefore, the vertex set of $\cbr{\bv^{\pi}\given \pi\in \Delta^{\cA}}$ should correspond to a subset of $\Delta^{\cA}$'s vertex set, which is the set of all pure strategies.

By Carath\'{e}odory's theorem, for any $\pi^*\in\Delta^{\cA}$, $\bv^{\pi^*}$ can be represented as the linear combination of $D + 1$ vertices, which further implies it can be written as a linear combination of $D + 1$ vectors in $\cbr{\bv^{\pi}\given \pi\text{ is a pure strategy}}$. Then, when $\bv^{\pi^*}=\sum_{k=1}^{D + 1} \lambda^k \bv^{\pi^k}$ with $\blambda\in\Delta^{D + 1}$ and $\pi^1,\pi^2,\dots$ are pure strategies, due to the linearity of $\bv^{\pi}$, we have
\begin{align*}
    \bv^{\pi^*}=\sum_{k=1}^{D + 1} \lambda^k \bv^{\pi^{k}}= \bv^{\sum_{k=1}^{D + 1} \lambda^k\pi^k}.
\end{align*}
Finally,
\begin{align*}
    &\max_{S\in \cS} \max_{\hat\ba_S\in\cA_S} \frac{1}{|S|}\sum_{i\in S} \EE_{\ba\sim \pi}\sbr{\cU_i\rbr{\hat\ba_S, \ba_{-S}} - \cU_i\rbr{\ba} }\\
    =&\max_{S\in \cS} \max_{\hat\ba_S\in\cA_S} \frac{1}{|S|}\sum_{i\in S} \EE_{\ba\sim \pi}\sbr{\cU_i\rbr{\hat\ba_{S\cap \cN(i)}, \ba_{-(S\cap \cN(i))}} - \cU_i\rbr{\ba} }\\
    =&\max_{S\in \cS} \max_{\hat\ba_S\in\cA_S} \frac{1}{|S|} \sum_{i\in S} v^{\pi^*}(S, i, \hat \ba_{S\cap \cN(i)})\\
    =&\max_{S\in \cS} \max_{\hat\ba_S\in\cA_S} \frac{1}{|S|} \sum_{i\in S} v^{\sum_{k=1}^{D + 1} \lambda^k\pi^k}(S, i, \hat \ba_{S\cap \cN(i)}).
\end{align*}
Hence, once $\pi^*$ satisfies \Cref{eq:MASE}, there exists a linear combination of $D + 1$ pure strategies also satisfying \Cref{eq:MASE}. Given $\cN(i)\leq \tw{\cG}$, we have $D\leq \sum_{S\in\cS} |S| \cdot A^{\tw{\cG}}$.
\end{proof}

\section{The Optimality of Dynamic Programming on Tree Decomposition}

This section shows that the dynamic programming (\emph{e.g.}, \Cref{eq:pi-dp}) will compute the optimality. Let $u_i(\ba_{\cN(i)})$ be the contribution of player $i$'s utility to the final objective. Recall that $\cN_i^S\coloneqq \cN(i)\cap S$. Then, in \Cref{eq:pi-dp},
\begin{align*}
    u_i(\ba_{\cN(i)})=-\sum_{\tau=1}^t \sum_{S\in\cS}\frac{1}{|S|} \sum_{\substack{i\in S}}~\sum_{\hat \ba_{\cN_i^S}\in\cA_{\cN_i^S}} ~ \mu^{(\tau)}\rbr{S, \hat \ba_{\cN_i^S}} \rbr{\cU_i\rbr{\hat \ba_{\cN_i^S}, \ba_{\cN(i)\setminus S}} - \cU_i(\ba_{\cN(i)})}
\end{align*}
at timestep $t$. We consider the following update rule in this section, which generalizes \Cref{eq:pi-dp} and \Cref{eq:mu-dp} (see the proof of \Cref{lemma:variation-d-max} and \Cref{lemma:variation-g-max}),
\begin{align}
    h(B, \ba_B)=&\sum_{\substack{i\in [N]\colon\\i\text{ assigned to }B}} u_i(\ba_{\cN(i)})+\sum_{B'\in C(B)} \max_{\substack{\ba_{B'}'\in\cA_{B'}\colon\\\ba_{B\cap B'}=\ba_{B\cap B'}'}} h(B', \ba_{B'}').\label{eq:general-dp}
\end{align}

In the following, we will show that \Cref{eq:general-dp} is optimal.
\begin{lemma}
\label{lemma:optimality}
    For any bag $B\in\cT$, let $\text{st}(B)\coloneqq \cbr{i}_{i\text{ assigned to }B}\cup \bigcup_{B'\in C(B)}\text{st}(B')$ be the set of players assigned to $B$ and bags in its subtree. Then, for any bag $B\in\cT$ and $\ba_B\in\cA_B$, we have
    \begin{align}
        h(B, \ba_B) = \max_{\ba_{-B}\in\cA_{-B}} \sum_{i\in \text{st}(B)} u_i(\ba_{\cN(i)}).\label{eq:optimality}
    \end{align}
\end{lemma}
The proof is postponed to the end of this section.
Note that for the root bag $B^r$, since $\text{st}(B^r)=[N]$, \Cref{lemma:optimality} implies that $\max_{\ba_{B^r}\in \cA_{B^r}} h(B^r, \ba_{B^r})=\max_{\ba\in\cA} \sum_{i=1}^N u_i(\ba_{\cN(i)})$. Therefore, we find the maximum of $\sum_{i=1}^N u_i(\ba_{\cN(i)})$, and the optimal joint action $\ba\in\cA$ can be extracted recursively.

Specifically, let 
\begin{equation}
\begin{split}
    &\ba_{B^r}^*=\argmax_{\ba_{B^r}\in \cA_{B^r}} h(B^r, \ba_{B^r})\\
    \forall B\in C(B^r),~~~& \ba_{B\setminus B^r}^*=\argmax_{\ba_{B\setminus B^r}\in\cA_{B\setminus B^r}} d\rbr{B, (\ba_{B\setminus B^r}, \ba^{(t+1)}_{B\cap B^r})}.
\end{split}
\label{eq:find-argmax-dp}
\end{equation}
We will do this recursively until we find the whole $\ba^*\in\cA$. The tie-breaking rule can be arbitrary, and we use the lexicographic order of joint actions for simplicity. Hence, we prove the optimality of the update rule \Cref{eq:general-dp}. \qed

\restate{lemma:optimality}

\begin{proof}
    For leaf bags $B$ ($C(B)=\emptyset$), for any joint action $\ba_B\in\cA_B$, we have
    \begin{align*}
        h(B, \ba_B)=\sum_{\substack{i\in [N]\colon\\i\text{ assigned to }B}} u_i(\ba_{\cN(i)})=&\sum_{i\in \text{st}(B)} u_i(\ba_{\cN(i)})\overset{(i)}{=} \max_{\ba_{-B}\in\cA_{-B}} \sum_{i\in \text{st}(B)} u_i(\ba_{\cN(i)}).
    \end{align*}
    $(i)$ is because $\cN(i)\subseteq B$ for any $i$ assigned to $B$ by definition. Additionally, since $B$ is a leaf bag, $\text{st}(B)=\cbr{i\in B\colon i\text{ assigned to }B}$.

    Then, for any bag $B$ with all of its children $B'\in C(B)$ satisfying \Cref{eq:optimality}, we have
    \begin{align*}
        h(B, \ba_B)=&\sum_{\substack{i\in [N]\colon\\i\text{ assigned to }B}} u_i(\ba_{\cN(i)})+\sum_{B'\in C(B)} \max_{\substack{\ba_{B'}'\in\cA_{B'}\colon\\\ba_{B\cap B'}=\ba_{B\cap B'}'}} h(B', \ba_{B'}')\\
        \overset{(i)}{=}& \sum_{\substack{i\in [N]\colon\\i\text{ assigned to }B}} u_i(\ba_{\cN(i)})+\sum_{B'\in C(B)} \max_{\substack{\ba_{B'}'\in\cA_{B'}\colon\\\ba_{B\cap B'}=\ba_{B\cap B'}'}} \max_{\ba_{-B'}'\in\cA_{-B'}} \sum_{i\in \text{st}(B')} u_i(\ba_{\cN(i)}')\\
        =&\sum_{\substack{i\in [N]\colon\\i\text{ assigned to }B}} u_i(\ba_{\cN(i)})+\sum_{B'\in C(B)} \max_{\substack{\ba'\in\cA\colon\\\ba_{B\cap B'}=\ba_{B\cap B'}'}} \sum_{i\in \text{st}(B')} u_i(\ba_{\cN(i)}').
    \end{align*}
    $(i)$ uses the induction hypothesis. By \Cref{item:td-property3}, for any $B'\in C(B)$ and $i\in \text{st}(B')$, $\cN(i)\cap \rbr{B\setminus B'}=\emptyset$. Because for any $i\in \text{st}(B')$, there must be a bag $B''$ in the subtree of $B'$ such that $\cN(i)\subseteq B''$, and \Cref{item:td-property3} will be violated if $\cN(i)\cap {B\setminus B'}\not=\emptyset$. Then, modifying the constraint $\ba_{B\cap B'}=\ba_{B\cap B'}'$ to $\ba_B=\ba_B'$ will not change the value of $u_i(\ba_{\cN(i)}')$ for any $i\in \text{st}(B')$. Hence,
    \begin{align*}
        h(B, \ba_B)=&\sum_{\substack{i\in [N]\colon\\i\text{ assigned to }B}} u_i(\ba_{\cN(i)})+\sum_{B'\in C(B)} \max_{\substack{\ba'\in\cA\colon\\\ba_B=\ba_B'}} \sum_{i\in \text{st}(B')} u_i(\ba_{\cN(i)}').
    \end{align*}
    Furthermore, by \Cref{item:td-property3}, for any $B',B''\in C(B)$ and $i'\in \text{st}(B'), i''\in \text{st}(B'')$, we have $\cN(i')\cap \cN(i'')\subseteq B$. Finally,
    \begin{align*}
        h(B, \ba_B)=&\sum_{\substack{i\in [N]\colon\\i\text{ assigned to }B}} u_i(\ba_{\cN(i)})+\max_{\substack{\ba'\in\cA\colon\\\ba_B=\ba_B'}} \sum_{B'\in C(B)} \sum_{i\in \text{st}(B')} u_i(\ba_{\cN(i)}')\\
        =&\max_{\ba_{-B}\in\cA_{-B}} \rbr{\sum_{\substack{i\in [N]\colon\\i\text{ assigned to }B}} u_i(\ba_{\cN(i)}) +  \sum_{B'\in C(B)} \sum_{i\in \text{st}(B')} u_i(\ba_{\cN(i)}')}\\
        =& \max_{\ba_{-B}\in\cA_{-B}} \sum_{i\in \text{st}(B)} u_i(\ba_{\cN(i)}).
    \end{align*}
    This completes the induction.\qedhere
\end{proof}

\section{Proof of \Cref{theorem:pi-no-regret}}
\label{section:pi-no-regret-proof}

\restate{theorem:pi-no-regret}

\begin{proof}
The proof of \Cref{theorem:pi-no-regret} can be decomposed into three steps.

Firstly, we show that without loss of generality, if FTPL with a fixed noise $\tilde \bn$ for all timesteps $t=1,2,\dots$ attains sublinear regret when the adversary is oblivious,\footnote{An oblivious adversary will choose all the utility functions at timestep $0$, while an adaptive adversary will choose the utility functions at timestep $t$ according to $\pi^{(1)},\pi^{(2)},\dots,\pi^{(t-1)}$.} then FTPL with independent noise vectors $\tilde \bn^{(t)}$ also attains the same regret confronting an adaptive adversary. The reduction to the oblivious setting is common in the literature \citep{agarwal2019learning-nonconvex,suggala2020online-nonconvex}, and we include it here for completeness.

Secondly, we will show that the regret of a fictitious algorithm $\pi^{(t+1)}\in \argmin_{\pi\in\Delta^{\cA}} \sum_{\tau=1}^{{\color{red}(t+1)}} F\rbr{\pi,\mu^{(\tau)}} + \inner{\tilde \bn^{(t+1)}}{\pi}$ is sublinear.

Finally, we will show that the regret of \Cref{eq:update-rule} and that of the fictitious algorithm are close.

\subsection{Fixed Noise Vector}

In this section, for completeness, we will show a reduction from an adaptive adversary to an oblivious adversary. For ease of representation, we will take correlator ($\pi$) as the \emph{no-regret learner}, and the deviator ($\mu$) as the \emph{adversary}.

An adaptive adversary determines the utility function at timestep $t$, which is $\mu^{(t)}$ in this section, according to our past strategies, $\pi^{(1)},\dots,\pi^{(t-1)}$. In contrast, an oblivious adversary determines all utility functions, \emph{i.e.}, $\mu^{(1)},\dots,\mu^{(T)}$, at the beginning (timestep $0$), such that $\mu^{(t)}$ is irrelevant to $\pi^{(1)},\dots,\pi^{(t-1)}$. In the following, we will show that a sublinear regret against an oblivious adversary implies a sublinear regret against an adaptive adversary.

Intuitively, when the random noise $\tilde \bn^{(1)},\dots,\tilde \bn^{(T)}$ are independent, $\pi^{(t)}$ only depends on $\mu^{(1)},\dots,\mu^{(t-1)}$, which is known to both the oblivious and adaptive adversary, due to the update rule \Cref{eq:update-rule}. Hence, an additional observation on $\pi^{(1)},\dots,\pi^{(t-1)}$ does not make adversary more powerful. Formally, we have the following lemma \citep[Lemma 4.1]{cesa2006prediction-Hedge}.

\begin{lemma}[Reformulation of Lemma 4.1 in \citet{cesa2006prediction-Hedge}]
\label{lemma:cesa-oblivious}
    Consider any randomized no-regret learner and the \emph{distribution} of the decision variable $\pi^{(t)}$ is fully determined by $\mu^{(1)},\dots,\mu^{(t-1)}$. Assume the no-regret learner's regret against any sequence of $\mu^{(1)},\dots,\mu^{(T)}$ generated by an oblivious adversary satisfies that
    \begin{align*}
        \underbrace{\EE_{\tilde \bn^{(1)},\dots,\tilde \bn^{(T)}}\sbr{\max_{\hat\pi\in\Delta^{\cA}} \sum_{t=1}^T F\rbr{\pi^{(t)},\mu^{(t)}} - F\rbr{\hat\pi,\mu^{(t)}}}}_{\text{Expected Regret}} \leq R.
    \end{align*}
    Then, for any sequence of $\mu^{(1)},\dots,\mu^{(T)}$ generated by an adaptive adversary and $\delta>0$, with probability at least $1-\delta$, we have
    \begin{align*}
        \max_{\hat\pi\in\Delta^{\cA}} \sum_{t=1}^T F\rbr{\pi^{(t)},\mu^{(t)}} - F\rbr{\hat\pi,\mu^{(t)}}\leq R + \sqrt{2 T\log \frac{1}{\delta}}.
    \end{align*}
\end{lemma}

It is easy to see that the distribution of $\pi^{(t)}$ generated by \Cref{eq:update-rule}, whose randomness is induced by $\tilde \bn^{(t)}$, is fully determined by $\mu^{(1)},\dots,\mu^{(t-1)}$, given $\tilde \bn^{(t)}$ is generated by a fixed distribution independently at each timestep.

Then, we will show that the expected regret of FTPL using independent noise vectors and FTPL with a fixed noise vector is the same, while facing an oblivious adversary.
\begin{align*}
    \EE_{\tilde \bn^{(1)},\dots,\tilde \bn^{(T)}}\sbr{\sum_{t=1}^T F\rbr{\pi^{(t)},\mu^{(t)}}}=&\sum_{t=1}^T \EE_{\tilde \bn^{(1)},\dots,\tilde \bn^{(T)}}\sbr{F\rbr{\pi^{(t)},\mu^{(t)}}}\\
    \overset{(i)}{=}&\sum_{t=1}^T \EE_{\tilde \bn^{(t)}} \sbr{F\rbr{\pi^{(t)},\mu^{(t)}}}.
\end{align*}
$(i)$ uses the fact that both $\pi^{(t)}$ and $\mu^{(t)}$ are independent of $\tilde \bn^{(1)},\dots,\tilde \bn^{(t-1)}$, when the adversary controlling $\mu$ is oblivious. Finally, the expectation $\EE_{\tilde\bn^{(1)}}\sbr{F\rbr{\pi^{(t)},\mu^{(t)}}}$ of
\begin{align*}
    \pi^{(t+1)}\in \argmin_{\pi\in\Delta^{\cA}} \sum_{\tau=1}^t F\rbr{\pi,\mu^{(\tau)}} - \inner{\tilde\bn^{(1)}}{\pi}
\end{align*}
is equal to $\EE_{\tilde \bn^{(t)}} \sbr{F\rbr{\pi^{(t)},\mu^{(t)}}}$, when $\bn^{(t)}$ and $\bn^{(1)}$ are sampled from an identical distribution.
In summary,
\begin{align*}
    &\text{fixed noise and an oblivious adversary (Expectation)}\\
    \Rightarrow&\text{independent noise and an oblivious adversary (Expectation)}\\
    \Rightarrow&\text{independent noise and an adaptive adversary (High Probability Bound)}.
\end{align*}
Hence, the problem reduces to proving sublinear regret against an oblivious adversary with a fixed noise vector for all timesteps.

\subsection{Low Regret with Accurate Prediction}

The discussion above suggests that we only need to show the sublinear regret against an oblivious adversary, when all timesteps share the same noise vector. In other words, we consider the following update rule,
\begin{align}
    \pi^{(t+1)}\in \argmin_{\pi\in\Delta^{\cA}} \sum_{\tau=1}^t F\rbr{\pi,\mu^{(\tau)}} - \inner{\tilde\bn}{\pi},
\end{align}
where $\tilde\bn,\tilde\bn^{(1)},\dots,\tilde\bn^{(T)}$ are identically distributed.

Next, we will show that if the regret minimizer can make the decision $\pi^{(t+1)}$ with an accurate prediction of $\mu^{(t+1)}$, then we can achieve sublinear regret. In particular, the decision variable at timestep $t+1$ is chosen according to \Cref{eq:pred-br}.
\begin{align}
    \pi^{(t+1)}\in\argmin_{\pi\in\Delta^{\cA}} \sum_{\tau=1}^{t{\color{red} +1}} F\rbr{\pi,\mu^{(\tau)}} - \inner{\tilde \bn}{\pi}.\label{eq:pred-br}
\end{align}
Actually, we can see that \Cref{eq:pred-br} is exactly the original update rule of $\pi^{(t+2)}$. Therefore, we will prove the following lemma in the sequel.
\begin{lemma}
\label{lemma:pred-br}
    Consider \Cref{eq:pred-br}. For any timestep $t=1,2,\dots$, $\mu^{(1)},\mu^{(2)},\dots$, and any $\hat\pi\in\Delta^{\cA}$, we have
    \begin{align}
        \sum_{\tau=1}^t \rbr{F\rbr{\pi^{(\tau{\color{red} +1})},\mu^{(\tau)}} - F\rbr{\hat\pi,\mu^{(\tau)}}}\leq \inner{\tilde \bn}{\pi^{(2)} - \hat\pi}.\label{lemma:pred-br-bound}
    \end{align}
    \proof
    We will prove the lemma by induction. When $t=1$, we have
    \begin{align*}
        &F\rbr{\pi^{(2)},\mu^{(1)}} - F\rbr{\hat\pi,\mu^{(1)}}\\
        =&\rbr{F\rbr{\pi^{(2)},\mu^{(1)}} - \inner{\tilde \bn}{\pi^{(2)}}} - \rbr{F\rbr{\hat\pi,\mu^{(1)}} - \inner{\tilde \bn}{\hat\pi}} + \inner{\tilde \bn}{\pi^{(2)} - \hat\pi}\\
        \overset{(i)}{\leq}& \inner{\tilde \bn}{\pi^{(2)} - \hat\pi}.
    \end{align*}
    $(i)$ is because $\pi^{(2)}\in \argmin_{\pi\in\Delta^{\cA}} F\rbr{\pi,\mu^{(1)}} - \inner{\tilde \bn}{\pi}$.

    Next, we will show that when \Cref{lemma:pred-br-bound} holds for $t=t_0$, then it also holds for $t=t_0+1$. For any $\hat\pi\in\Delta^{\cA}$, we have
    \begin{align*}
        &\sum_{\tau=1}^{t_0+1} \rbr{F\rbr{\pi^{(\tau+1)},\mu^{(\tau)}} - F\rbr{\hat\pi,\mu^{(\tau)}}}\\
        =& \sum_{\tau=1}^{t_0+1} F\rbr{\pi^{(\tau+1)},\mu^{(\tau)}} - \rbr{\sum_{\tau=1}^{t_0+1} F\rbr{\hat\pi,\mu^{(\tau)}} - \inner{\tilde \bn}{\hat\pi}} - \inner{\tilde \bn}{\hat\pi} \\
        \overset{(i)}{\leq}& \sum_{\tau=1}^{t_0+1} F\rbr{\pi^{(\tau+1)},\mu^{(\tau)}} - \rbr{\sum_{\tau=1}^{t_0+1} F\rbr{\pi^{(t_0+2)},\mu^{(\tau)}} - \inner{\tilde \bn}{\pi^{(t_0+2)}}} - \inner{\tilde \bn}{\hat\pi} \\
        =& \sum_{\tau=1}^{t_0} \rbr{F\rbr{\pi^{(\tau+1)},\mu^{(\tau)}} - F\rbr{\pi^{(t_0+2)},\mu^{(\tau)}}} + \inner{\tilde \bn}{\pi^{(t_0+2)} - \hat\pi}.
    \end{align*}
    $(i)$ is because $\pi^{(t_0+2)}\in \argmin_{\pi\in\Delta^{\cA}} \sum_{\tau=1}^{t_0+1} F\rbr{\pi,\mu^{(\tau)}} - \inner{\tilde \bn}{\pi}$. Next, by setting $\hat\pi=\pi^{(t_0+2)}$ in the induction hypothesis, we have
    \begin{align*}
        &\sum_{\tau=1}^{t_0} \rbr{F\rbr{\pi^{(\tau+1)},\mu^{(\tau)}} - F\rbr{\pi^{(t_0+2)},\mu^{(\tau)}}} + \inner{\tilde \bn}{\pi^{(t_0+2)} - \hat\pi}\\
        \leq& \inner{\tilde \bn}{\pi^{(2)} - \pi^{(t_0+2)}} + \inner{\tilde \bn}{\pi^{(t_0+2)} - \hat\pi}\\
        =& \inner{\tilde \bn}{\pi^{(2)} - \hat\pi}.\qedhere
    \end{align*}
\end{lemma}

\subsection{Sublinear Variation}

In this section, we will show that the regret of FTPL with/without a prediction of $\mu^{(t+1)}$ is close. Formally, for any $\hat\pi\in\Delta^{\cA}$, we have
\begin{align*}
    &\sum_{t=1}^T \rbr{F\rbr{\pi^{(t)},\mu^{(t)}} - F\rbr{\hat\pi,\mu^{(t)}}}\\
    =&\sum_{t=1}^T \rbr{F\rbr{\pi^{(t+1)},\mu^{(t)}} - F\rbr{\hat\pi,\mu^{(t)}}} + \sum_{t=1}^T \rbr{F\rbr{\pi^{(t)},\mu^{(t)}} - F\rbr{\pi^{(t+1)},\mu^{(t)}}}\\
    \overset{(i)}{\leq}& \inner{\tilde \bn}{\pi^{(2)} - \hat\pi} + \sum_{t=1}^T \rbr{F\rbr{\pi^{(t)},\mu^{(t)}} - F\rbr{\pi^{(t+1)},\mu^{(t)}}}.
\end{align*}
$(i)$ uses \Cref{lemma:pred-br}.

Moreover, since $\cU_i(\ba)\in [0, 1]$ for any $i\in [N]$ and $\ba\in\cA$, we have
\begin{align*}
    \max_{\substack{\pi\in\Delta^{\cA},\\\mu\in\Delta^{\bigtimes_{S\in\cS} \cA_S}}} \abr{F\rbr{\pi, \mu}}\leq 1.
\end{align*}
Then,
\begin{align*}
    \EE\sbr{F\rbr{\pi^{(t)},\mu^{(t)}} - F\rbr{\pi^{(t+1)},\mu^{(t)}}}\leq \Pr_{\tilde \bn}\rbr{\pi^{(t)}\not=\pi^{(t+1)}}.
\end{align*}
Hence, we only need to lower bound $\Pr_{\tilde \bn}\rbr{\pi^{(t)}=\pi^{(t+1)}}$ in the sequel. Recall that $\pi^{(t)}$ is a pure strategy with $\pi^{(t)}\rbr{\ba^{(t)}}=1$ for some joint action $\ba^{(t)}\in\cA$. For any bag $B\in\cT$, let $fa(B)$ denote its father ($fa(B)=\emptyset$ if $B$ is the root). Then,
\begin{align*}
    \Pr_{\tilde \bn}\rbr{\pi^{(t)}=\pi^{(t+1)}}=&\Pr_{\tilde \bn}\rbr{\ba^{(t)}=\ba^{(t+1)}}\\
    \overset{(i)}{=}&\prod_{B\in\cT} \Pr_{\tilde \bn}\rbr{\ba_{B\setminus fa(B)}^{(t)}=\ba_{B\setminus fa(B)}^{(t+1)}\biggiven \ba_{fa(B)}^{(t)}=\ba_{fa(B)}^{(t+1)}}.
\end{align*}
According to \Cref{eq:find-argmax-dp-pi}, each bag $B\in\cT$ only determines $\ba_{B\setminus fa(B)}^{(t)}$. Since $\bn^{(t)}(B,\cdot)$ is sampled independently for every bag $B$, it follows that $\ba_{B'\setminus B}^{(t)},\ba_{B''\setminus B}^{(t)}$ are independent for $B',B''\in C(B)$ with $B'\not=B''$, by \Cref{item:td-property3}. Hence, $(i)$ holds.

For any $B\in\cT$, to lower bound $\Pr_{\tilde \bn}\rbr{\ba_{B\setminus fa(B)}^{(t)}=\ba_{B\setminus fa(B)}^{(t+1)}\biggiven \ba_{fa(B)}^{(t)}=\ba_{fa(B)}^{(t+1)}}$, we will first get its lower bound while further conditioning on $\ba^{(t)}_{B\cup fa(B)}$'s value and $\bn(B,\cdot)$'s value.

Then, $\Pr_{\tilde \bn}\rbr{\ba_{B\setminus fa(B)}^{(t)}=\ba_{B\setminus fa(B)}^{(t+1)}\biggiven \ba_{fa(B)}^{(t)}=\ba_{fa(B)}^{(t+1)}}$ is equal to this conditioned probability integrating over all possible values of $\ba^{(t)}_{B\cup fa(B)}$ and $\bn(B,\cdot)$. Formally, we want to lower bound
\begin{align*}
    p_B^{(t)}(\ba_{B\cup fa(B)}', \bx)\coloneqq\Pr\Big(&\ba_{B\setminus fa(B)}^{(t+1)}=\ba_{B\setminus fa(B)}'\Biggiven \ba_{fa(B)}^{(t)}=\ba_{fa(B)}^{(t+1)}~,~\ba_{B\cup fa(B)}^{(t)}=\ba_{B\cup fa(B)}'~,~\\
    &\text{ and } \forall \ba_{B\setminus fa(B)}\in\cA_{B\setminus fa(B)}\setminus\cbr{\ba_{B\setminus fa(B)}'},\\
    &n(B, (\ba'_{B\cap fa(B)}, \ba_{B\setminus fa(B)}))=x((\ba'_{B\cap fa(B)},\ba_{B\setminus fa(B)}))\Big)
\end{align*}
for any $\ba_{B\cup fa(B)}'\in\cA_{B\cup fa(B)}$ and $\bx\in\RR^{\cA_B}$. Then,
\begin{align*}
    &\Pr_{\tilde \bn}\rbr{\ba_{B\setminus fa(B)}^{(t)}=\ba_{B\setminus fa(B)}^{(t+1)}\biggiven \ba_{fa(B)}^{(t)}=\ba_{fa(B)}^{(t+1)}}\geq \inf_{\substack{\ba_{B\cup fa(B)}'\in\cA_{B\cup fa(B)},\\\bx\in\RR^{\cA_B}}} p_B^{(t)}(\ba_{B\cup fa(B)}', \bx),
\end{align*}
since $\Pr_{\tilde \bn}\rbr{\ba_{B\setminus fa(B)}^{(t)}=\ba_{B\setminus fa(B)}^{(t+1)}\biggiven \ba_{fa(B)}^{(t)}=\ba_{fa(B)}^{(t+1)}}$ is equal to $p_B^{(t)}(\ba_{B\cup fa(B)}', \bx)$ integrating over $\ba_{B\cup fa(B)}'$ and $\bx$.

Since $\ba_{B\setminus fa(B)}^{(t)}=\ba_{B\setminus fa(B)}'$, for any $\ba_{B\setminus fa(B)}\in \cA_{B\setminus fa(B)}\setminus\cbr{\ba_{B\setminus fa(B)}'}$, we have
\begin{align*}
    &d^{(t)}\rbr{B, \ba_B'}\leq d^{(t)}\rbr{B, (\ba_{B\cap fa(B)}', \ba_{B\setminus fa(B)})}.
\end{align*}
This can be equivalently written as
\begin{align*}
    n\rbr{B, \ba_B'}\geq& \rbr{d^{(t)}\rbr{B, \ba_B'} + n\rbr{B, \ba_B'}} -d^{(t)}\rbr{B, (\ba_{B\cap fa(B)}', \ba_{B\setminus fa(B)})}.
\end{align*}
Then, $\ba_{B\setminus fa(B)}^{(t+1)}=\ba_{B\setminus fa(B)}'$ is equivalent to
\begin{align}
    n\rbr{B, \ba_B'}\geq& \rbr{d^{(t+1)}\rbr{B, \ba_B'} + n\rbr{B, \ba_B'}} - d^{(t+1)}\rbr{B, (\ba_{B\cap fa(B)}', \ba_{B\setminus fa(B)})}\label{eq:t-plus-1-optimal}\\
    =& \rbr{d^{(t)}\rbr{B, \ba_B'} + n\rbr{B, \ba_B'}} -d^{(t)}\rbr{B, (\ba_{B\cap fa(B)}', \ba_{B\setminus fa(B)})}\notag\\
    &+ \rbr{d^{(t)}\rbr{B, (\ba_{B\cap fa(B)}', \ba_{B\setminus fa(B)})} - d^{(t+1)}\rbr{B, (\ba_{B\cap fa(B)}', \ba_{B\setminus fa(B)})}} \notag\\
    & + \rbr{d^{(t+1)}\rbr{B, \ba_B'} - d^{(t)}\rbr{B, \ba_B'}}.\notag
\end{align}
for any $\ba_{B\setminus fa(B)}\in \cA_{B\setminus fa(B)}$. In \Cref{lemma:variation-d-max}, we show that the variation of $\bd$ is bounded by $1$. Therefore,
\begin{align*}
    n\rbr{B, \ba_B'}\geq \rbr{d^{(t)}\rbr{B, \ba_B'} + n\rbr{B, \ba_B'}} -d^{(t)}\rbr{B, (\ba_{B\cap fa(B)}', \ba_{B\setminus fa(B)})} + 2
\end{align*}
implies \Cref{eq:t-plus-1-optimal}.
\begin{lemma}
\label{lemma:variation-d-max}
Consider the update rule \Cref{eq:update-rule}. For any timestep $t=1,2,\dots,T$, bag $B\in\cT$, joint action $\ba_B\in\cA_B$, and noise $\tilde \bn\in\RR^{\bigtimes_{B\in\cT} \cA_B}$, we have
\begin{align*}
    \abr{d^{(t+1)}(B, \ba_B) - d^{(t)}(B, \ba_B)}\leq 1.
\end{align*}
\end{lemma}
The proof is postponed to the end of this section. Let 
\begin{align*}
    w=\max_{\ba_{B\setminus fa(B)}\in \cA_{B\setminus fa(B)}\setminus\cbr{\ba_{B\setminus fa(B)}'}}\allowbreak \rbr{d^{(t)}\rbr{B, \ba_B'} + n\rbr{B, \ba_B'}} - d^{(t)}\rbr{B, (\ba_{B\cap fa(B)}', \ba_{B\setminus fa(B)})}.
\end{align*}
Note that $w$ only depends on $\mu^{(1)},\dots,\mu^{(t-1)}$ and $\bx$. Then,
\begin{align*}
    p_B^{(t)}(\ba_{B\cup fa(B)}', \bx)\geq&\Pr\rbr{n\rbr{B, \ba_B'}\geq w + 2 \given n\rbr{B, \ba_B'}\geq w}\\
    =& \frac{\Pr\rbr{n\rbr{B, \ba_B'}\geq w + 2}}{\Pr\rbr{n\rbr{B, \ba_B'}\geq w}}\\
    \overset{(i)}{=}& \frac{\exp\rbr{-\eta\rbr{w+ 2}}}{\exp\rbr{-\eta w}}\\
    =&\exp\rbr{-2\eta}.
\end{align*}
$(i)$ is because $n\rbr{B, \ba_B'}\sim \Exp(\eta)$. Finally, by union bound,
\begin{align*}
    \Pr_{\tilde \bn}\rbr{\pi^{(t)}=\pi^{(t+1)}}\geq& 1 - \sum_{B\in\cT} \rbr{1-p_B^{(t)}(\ba_{B\cup fa(B)}', \bx)}\\
    \geq& 1 - \sum_{B\in\cT} \rbr{1-\exp\rbr{-2\eta}}\\
    \geq& 1-\sum_{B\in\cT} 2\eta\\
    =&1-2\eta\abr{\cT}.
\end{align*}
Therefore,
\begin{align*}
    &\EE\sbr{\sum_{t=1}^T F\rbr{\pi^{(t)},\mu^{(t)}} - F\rbr{\hat\pi,\mu^{(t)}}}\\
    \leq& \EE\sbr{\inner{\tilde \bn}{\pi^{(2)} - \hat\pi}} + \sum_{t=1}^T \EE\sbr{F\rbr{\pi^{(t)},\mu^{(t)}} - F\rbr{\pi^{(t+1)},\mu^{(t)}}}\\
    \leq&  \EE\sbr{\inner{\tilde \bn}{\pi^{(2)} - \hat\pi}} + 2\eta\abr{\cT} T.
\end{align*}
Since $\exp(x)\geq 1+x$ for any $x\in\RR$, $\rbr{1-\exp\rbr{-2\eta\abr{\cT}}}\leq 2\eta\abr{\cT}$. Additionally,
\begin{align*}
    \EE\sbr{\inner{\tilde \bn}{\pi^{(2)} - \hat\pi}}\overset{(i)}{\leq} \EE\sbr{\nbr{\tilde \bn}_\infty\cdot \nbr{\pi^{(2)} - \hat\pi}_1}\leq& 2\EE\sbr{\nbr{\tilde \bn}_\infty}\\
    \leq& 2\sum_{B\in\cT} \max_{\ba_B\in\cA_B} n(B, \ba_B)\\
    \overset{(i)}{\leq}& 2\sum_{B\in\cT} \frac{1+\log \abr{\cA_B}}{\eta}.
\end{align*}
$(i)$ is by \holder. $(ii)$ is because the expectation of the maximum of $n$ i.i.d. random variable sampled from $\Exp(\eta)$ is upper bounded by $\frac{1+\log n}{\eta}$ \citep{agarwal2019learning-nonconvex}. Furthermore, $\log \abr{\cA_B}\leq |B|\cdot\log A\leq \rbr{\tw{\cG}+1}\log A$. Hence,
\begin{align*}
    \EE\sbr{\sum_{t=1}^T F\rbr{\pi^{(t)},\mu^{(t)}} - F\rbr{\hat\pi,\mu^{(t)}}}\leq 2\abr{\cT} \frac{1+\rbr{\tw{\cG}+1}\log A}{\eta} + 2\eta\abr{\cT} T.
\end{align*}
\qed
\end{proof}

\subsection{Proof of Auxiliary Lemmas}

\restate{lemma:variation-d-max}

\begin{proof}
    Recall \Cref{lemma:optimality}. We can add $|\cT|$ players as the noise player, each assigned to a bag in $\cT$, with $u_i(\ba_B)=n(B, \ba_B)$ so that $\cN(i)=B$. Recall that $u_i$ is the contribution of player $i$ to the objective function $F$. Then, by \Cref{lemma:optimality}, for any $B\in\cT$ and $\ba_B\in\cA_B$, we have
    \begin{align*}
        d^{(t)}(B,\ba_B)=\min_{\ba_{-B}\in\cA_{-B}} \sum_{i\in \text{st}(B)} u_i^{(t)}(\ba_{\cN(i)}),
    \end{align*}
    where
    \begin{align*}
        &u_i^{(t)}(\ba_{\cN(i)})=-\sum_{\tau=1}^t \sum_{\substack{S\in\cS\colon\\i\in S}} ~ \frac{1}{|S|} \sum_{\hat\ba_{\cN_i^S}\in\cA_{\cN_i^S}} \mu^{(\tau)}\rbr{S, \hat\ba_{\cN_i^S}} \rbr{\cU_i\rbr{\hat\ba_{\cN_i^S}, \ba_{\cN(i)\setminus S}} - \cU_i(\ba_{\cN(i)})}\text{ for }i\in [N]\\
        &u_i^{(t)}(\ba_{\cN(i)})=n(B, \ba_{\cN(i)})\text{ for $i$ as the noise player assigned to bag } B.
    \end{align*}
    For any noise player, we can see that $u_i^{(t)}(\ba_{\cN(i)}) - u_i^{(t+1)}(\ba_{\cN(i)})=0$. For any $i\in [N]$,
    \begin{align*}
        &\abr{u_i^{(t)}(\ba_{\cN(i)}) - u_i^{(t+1)}(\ba_{\cN(i)})}\\
        =&\abr{\sum_{\substack{S\in\cS\colon\\i\in S}} ~ \frac{1}{|S|}\sum_{\hat\ba_{\cN_i^S}\in\cA_{\cN_i^S}} \mu^{(t+1)}\rbr{S, \hat\ba_{\cN_i^S}} \rbr{\cU_i\rbr{\hat\ba_{\cN_i^S}, \ba_{\cN(i)\setminus S}} - \cU_i(\ba_{\cN(i)})}}\\
        \leq&\sum_{\substack{S\in\cS\colon\\i\in S}} ~ \frac{1}{|S|}\sum_{\hat\ba_{\cN_i^S}\in\cA_{\cN_i^S}} \mu^{(t+1)}\rbr{S, \hat\ba_{\cN_i^S}} \abr{\cU_i\rbr{\hat\ba_{\cN_i^S}, \ba_{\cN(i)\setminus S}} - \cU_i(\ba_{\cN(i)})}\\
        \overset{(i)}{\leq}& \sum_{\substack{S\in\cS\colon\\i\in S}} ~ \frac{1}{|S|}\sum_{\hat\ba_{\cN_i^S}\in\cA_{\cN_i^S}} \mu^{(t+1)}\rbr{S, \hat\ba_{\cN_i^S}}.
    \end{align*}
    $(i)$ is because $\cU_i(\ba_{\cN(i)})\in [0,1]$ for any $\ba_{\cN(i)}\in\cA_{\cN(i)}$.
    
    Recall that by definition, $\ba_{-B}^{(t)}=\argmax_{\ba_{-B}\in\cA_{-B}} \sum_{i\in \text{st}(B)} u_i^{(t)}(\ba_{\cN(i)})$. Then,
    \begin{align*}
        &d^{(t)}(B,\ba_B)\\
        =&\sum_{i\in \text{st}(B)} u_i^{(t)}\rbr{(\ba_{\cN(i)\cap B}, \ba_{\cN(i)\setminus B}^{(t)})}\\
        \geq& \sum_{i\in \text{st}(B)} u_i^{(t)}\rbr{(\ba_{\cN(i)\cap B}, \ba_{\cN(i)\setminus B}^{(t+1)})}\\
        =& \sum_{i\in \text{st}(B)} u_i^{(t+1)}\rbr{(\ba_{\cN(i)\cap B}, \ba_{\cN(i)\setminus B}^{(t+1)})} +\rbr{ u_i^{(t+1)}\rbr{(\ba_{\cN(i)\cap B}, \ba_{\cN(i)\setminus B}^{(t+1)})} - u_i^{(t)}\rbr{(\ba_{\cN(i)\cap B}, \ba_{\cN(i)\setminus B}^{(t+1)})} }\\
        \geq& \sum_{i\in \text{st}(B)} \rbr{u_i^{(t+1)}\rbr{(\ba_{\cN(i)\cap B}, \ba_{\cN(i)\setminus B}^{(t+1)})} - \sum_{\substack{S\in\cS\colon\\i\in S}} ~ \frac{1}{|S|}\sum_{\hat\ba_{\cN_i^S}\in\cA_{\cN_i^S}} \mu^{(t+1)}\rbr{S, \hat\ba_{\cN_i^S}}}\\
        =& d^{(t+1)}(B,\ba_B) - \sum_{S\in\cS} ~\frac{1}{|S|} \sum_{i\in S\cap \text{st}(B)}~ \sum_{\hat\ba_{\cN_i^S}\in\cA_{\cN_i^S}} \mu^{(t+1)}\rbr{S, \hat\ba_{\cN_i^S}}\\
        \geq& d^{(t+1)}(B,\ba_B) - 1.
    \end{align*}
    Similarly, we can get the upper bound that $d^{(t)}(B,\ba_B)\leq d^{(t+1)}(B,\ba_B) + 1$. Hence, the proof is completed.
\end{proof}

\section{Efficient Update of $\mu$}
\label{section:mu-update}

This section provides the omitted details regarding the update procedure for $\mu$ and presents the complete proof of \Cref{theorem:mu-no-regret}.

\subsection{Efficient Update of $\mu$}

The procedure for updating $\mu$ closely parallels that of $\pi$. Specifically, we iterate over all coalitions $S\in\cS$ and, for each $S$, determine the optimal action $\ba_S\in\cA_S$. To achieve this, we maintain a dynamic programming vector $\bg_S\in\RR^{\bigtimes_{B\in\cT}\cA_B}$ for each $S\in\cS$, which is updated according to
\begin{equation}
\begin{split}
    g_S^{(t+1)}(B, \hat\ba_B) =& \frac{1}{|S|} \sum_{\tau=1}^t \sum_{\substack{i\in S\colon\\i\text{ assigned to }B}} {\color{blue}\sum_{\ba\in\cA}} \pi^{(\tau)}(\ba) \rbr{\cU_i\rbr{\hat\ba_{B\cap S}, \ba_{B\setminus S}} - \cU_i\rbr{\ba_B}} \notag\\
    &+\sum_{B'\in C(B)} \max_{\substack{\hat\ba_{B'}'\in\cA_{B'}\colon\\\hat\ba_{B\cap B'}=\hat\ba'_{B\cap B'}}} g^{(t+1)}(B', \hat \ba_{B'}') + m^{(t+1)}(B, \hat \ba_B).
\end{split}
\label{eq:mu-dp-1}
\end{equation}
At first sight, the {\color{blue} $\sum_{\ba\in\cA}$} appears computationally prohibitive, since $\cA$ is exponentially large. Fortunately, the update becomes tractable once we recall that $\pi^{(\tau)}$ is always a pure strategy for $\tau\geq 1$. Denote by $\ba^{(\tau)}$ the joint action selected by $\pi^{(\tau)}$. Then \Cref{eq:mu-dp-1} simplifies to
\begin{equation}
\begin{split}
    g_S^{(t+1)}(B, \hat\ba_B) =& \frac{1}{|S|} \sum_{\tau=1}^t \sum_{\substack{i\in S\colon\\i\text{ assigned to }B}} \rbr{\cU_i\rbr{\hat\ba_{B\cap S}, {\color{blue} \ba_{B\setminus S}^{(\tau)}}} -  \cU_i\rbr{{\color{blue}\ba_B^{(\tau)}}}} \\
    &+\sum_{B'\in C(B)} \max_{\substack{\hat\ba_{B'}'\in\cA_{B'}\colon\\\hat\ba_{B\cap B'}=\hat\ba'_{B\cap B'}}} g^{(t+1)}(B', \hat \ba_{B'}') + m_S^{(t+1)}(B, \hat \ba_B).
\end{split}
\label{eq:mu-dp}
\end{equation}
After completing the dynamic programming updates, we focus on the root bag $B^r$ of the tree decomposition. The selected coalition is then $S^{(t+1)}=\allowbreak\argmax_{S\in\cS} \allowbreak\max_{\hat \ba_{B^r}\in\cA_{B^r}}\allowbreak g_S^{(t+1)} (B^r,\hat \ba_{B^r})$. Next, we apply the reconstruction procedure in \Cref{eq:find-argmax-dp-pi} on $\bg_{S^{(t+1)}}^{(t+1)}$ to extract a joint action $\hat\ba^{(t+1)}\in\cA$. Finally, we update $\mu^{(t+1)}\rbr{S^{(t+1)}, \hat\ba^{(t+1)}}=1$.

This procedure ensures that $\mu$ can be updated efficiently while maintaining consistency with the tree decomposition structure. Analogous to the update of $\pi$, the regret of this process can be bounded.

\subsection{Proof of \Cref{theorem:mu-no-regret}}
\label{section:proof-mu-no-regret}

\restate{theorem:mu-no-regret}

\begin{proof}
    The proof of \Cref{theorem:mu-no-regret} is similar to that of \Cref{theorem:pi-no-regret}. By using a similar argument as the proof of \Cref{theorem:pi-no-regret}, for any $\hat\mu\in\Delta^{\bigtimes_{S\in\cS} \cA_S}$, we have
    \begin{align*}
        &\sum_{t=1}^T \rbr{F\rbr{\pi^{(t)},\hat \mu} - F\rbr{\pi^{(t)},\mu^{(t)}}}\\
        \leq& \inner{\tilde \bbm}{\hat\mu - \mu^{(2)}} + \sum_{t=1}^T \rbr{F\rbr{\pi^{(t)},\mu^{(t+1)}} - F\rbr{\pi^{(t)},\mu^{(t)}}}.
    \end{align*}
    Next, by introducing the counterpart of \Cref{lemma:variation-d-max} in the following, the rest of the proof follows that of \Cref{theorem:pi-no-regret}.
    \begin{lemma}
        \label{lemma:variation-g-max}
        Consider the update rule \Cref{eq:update-rule}. For any timestep $t=1,2,\dots,T$, bag $B\in\cT$, joint action $\hat\ba_B\in\cA_B$, coalition $S\in\cS$, and noise $\tilde \bbm\in\RR^{\bigtimes_{B\in\cT} \cA_B}$, we have
        \begin{align*}
            \abr{g_S^{(t+1)}(B, \hat\ba_B) - g_S^{(t)}(B, \hat\ba_B)}\leq 1.
        \end{align*}
\end{lemma}
The proof is postponed to the end of this section.
    Let $S^{(t)},\hat\ba^{(t)}$ denote the coalition and action the deviator picks at timestep $t$, \emph{i.e.}, $\mu^{(t)}\rbr{S^{(t)},\hat\ba^{(t)}}=1$. Then,
    \begin{align*}
        &\EE\sbr{F\rbr{\pi^{(t)},\mu^{(t+1)}} - F\rbr{\pi^{(t)},\mu^{(t)}}} \\
        \leq& \Pr_{\tilde\bbm}\rbr{\mu^{(t+1)}\not=\mu^{(t)}}\\
        =&\Pr_{\tilde\bbm}\rbr{S^{(t)}=S^{(t+1)}} \prod_{B\in\cT} \Pr_{\tilde \bbm}\rbr{\hat\ba_{B\setminus fa(B)}^{(t)}=\hat\ba_{B\setminus fa(B)}^{(t+1)}\biggiven \hat\ba_{fa(B)}^{(t)}=\hat\ba_{fa(B)}^{(t+1)}, S^{(t)}=S^{(t+1)}}\\
        \overset{(i)}{\leq}& 1-\exp\rbr{2\eta |\cT|}.
    \end{align*}
    $(i)$ is because choosing $S^{(t)}$ is equivalent to adding a new player in the root bag $B^r$, whose action is to select the coalition. Finally, for any $\delta>0$, with probability at least $1-\delta$,
    \begin{align*}
        &\sum_{t=1}^T \rbr{F\rbr{\pi^{(t)},\hat \mu} - F\rbr{\pi^{(t)},\mu^{(t)}}}\\
        \leq& 2\abr{\cT} \frac{1+\rbr{\tw{\cG}+1}\log A}{\eta} + 2\eta\abr{\cT} T + \sqrt{2T\log\frac{1}{\delta}}.\qedhere
    \end{align*}
\end{proof}

\subsection{Proof of Auxiliary Lemmas}

\restate{lemma:variation-g-max}
\begin{proof}
    For any $S\in\cS$, the upper bound of $\abr{g_S^{(t+1)}(B, \hat\ba_B) - g_S^{(t)}(B, \hat\ba_B)}$ can be obtained similarly to the proof of \Cref{lemma:variation-d-max} by choosing
    \begin{align*}
        &u_i^{(t)}(\ba_{\cN(i)})=\frac{1}{|S|} \sum_{\tau=1}^t \sum_{\substack{i\in S\colon\\i\text{ assigned to }B}} \sum_{\ba\in\cA} \pi^{(\tau)}(\ba) \rbr{\cU_i\rbr{\hat\ba_{B\cap S}, \ba_{B\setminus S}} - \cU_i\rbr{\ba_B}}\text{ for }i\in [N]\\
        &u_i^{(t)}(\ba_{\cN(i)})=m_S(B, \ba_{\cN(i)})\text{ for $i$ as the noise player assigned to bag } B.
    \end{align*}
    Then,
    \begin{align*}
        \abr{u_i^{(t)}(\ba_{\cN(i)}) - u_i^{(t+1)}(\ba_{\cN(i)})}
        =&\abr{\frac{1}{|S|} \sum_{\substack{i\in S\colon\\i\text{ assigned to }B}} \sum_{\ba\in\cA} \pi^{(t+1)}(\ba) \rbr{\cU_i\rbr{\hat\ba_{B\cap S}, \ba_{B\setminus S}} - \cU_i\rbr{\ba_B}}}\\
        \overset{(i)}{\leq}& \frac{1}{|S|} \sum_{\substack{i\in S\colon\\i\text{ assigned to }B}} \sum_{\ba\in\cA} \pi^{(t+1)}(\ba).
    \end{align*}
    $(i)$ is by the fact that $\cU_i\rbr{\hat\ba_{B\cap S}, \ba_{B\setminus S}}, \cU_i\rbr{\ba_B}\in [0, 1]$. The rest of the proof follows that of \Cref{lemma:variation-d-max}, and thus we complete the proof.
\end{proof}

\section{Additional Experimental Results}
\label{section:experiment-details}

In this section, we detail our experimental setup and report additional results for two further games: the Chicken game \citep{bergstrom1998evolution-chicken} and Pigou’s network \citep{pigou1920economics-pigou}. All experiments are conducted on a 13th Gen Intel(R) Core(TM) i7-13700K @ 3.40\,GHz. Error bars for \textcolor[HTML]{6c3483}{MASE} and \textcolor[HTML]{B85C38}{FTPL} indicate $\pm 1\sigma$ over $100$ random seeds ($0,1,\dots,99$). Across all experiments, we set the learning rate to $\eta=0.01$ and run for $T=10,000$ timesteps.

\subsection{Experimental Details}

For all baselines, we run each algorithm independently for each player, thus the average strategy converges to a CCE \citep{hazan2016introduction}. For \textcolor[HTML]{ff9800}{FTRL}, \textcolor[HTML]{a31214}{Hedge}, and \textcolor[HTML]{25597a}{OMD}, each player $i\in [N]$ is initialized with a uniform distribution $\pi_i^{(1)}$ over all actions. \textcolor[HTML]{6c3483}{MASE} and \textcolor[HTML]{B85C38}{FTPL} are initialized with a pure strategy chosen uniformly at random from all pure strategies.

\subsection{Utility Functions}

This subsection specifies the utility functions for all four games.

\paragraph{Prisoner’s Dilemma.} The utility matrix is shown in \Cref{table:PD-appendix}. If both prisoners confess, they receive reduced sentences. If one confesses while the other defects, the confessor is imprisoned and the defector is released immediately. If both defect, both are imprisoned for longer than in the mutual-confession case to penalize dishonesty.

\begin{table}[h]
    \centering
    \begin{tabular}{|c|c|c|}
    \hline
     & \textbf{Confess (C)} & \textbf{Defect (D)} \\
    \hline
    \textbf{Confess (C)} & $(0.6,0.6)$ & $(0,1)$ \\
    \hline
    \textbf{Defect (D)} & $(1,0)$ & $(0.2,0.2)$ \\
    \hline
    \end{tabular}
    \caption{Utility matrix of Prisoner’s Dilemma. Each entry $(a,b)$ denotes the payoffs to the row player ($a$) and the column player ($b$).}
    \label{table:PD-appendix}
\end{table}

\paragraph{Stag Hunt.} The utility matrix is shown in \Cref{table:SH-appendix}. A stag yields a higher reward, but it can only be hunted successfully if both players choose Stag; a solo stag attempt yields nothing. A hare provides a smaller payoff but can be secured by a single player.

\begin{table}[h]
    \centering
    \begin{tabular}{|c|c|c|}
    \hline
     & \textbf{Stag (S)} & \textbf{Hare (H)} \\
    \hline
    \textbf{Stag (S)} & $(1, 1)$ & $(0.1,0.8)$ \\
    \hline
    \textbf{Hare (H)} & $(0.8,0.1)$ & $(0.5,0.5)$ \\
    \hline
    \end{tabular}
    \caption{Utility matrix of the Stag Hunt. Each entry $(a,b)$ denotes the payoff of the row player ($a$) and the column player ($b$).}
    \label{table:SH-appendix}
\end{table}

\paragraph{Chicken Game.} Two drivers head toward each other and can either swerve or go straight. If one goes straight while the other swerves, the swerving player "loses". If both go straight, they crash.

\begin{table}[h]
    \centering
    \begin{tabular}{|c|c|c|}
    \hline
     & \textbf{Swerve (Sw)} & \textbf{Straight (St)} \\
    \hline
    \textbf{Swerve (Sw)} & $(5/6, 5/6)$ & $(2/3, 1)$ \\
    \hline
    \textbf{Straight (St)} & $(1,2/3)$ & $(0,0)$ \\
    \hline
    \end{tabular}
    \caption{Utility matrix of the Chicken game. Each entry $(a,b)$ denotes the payoff of the row player ($a$) and the column player ($b$).}
    \label{table:chicken-appendix}
\end{table}

\paragraph{Pigou Network.} We use a three-player variant of Pigou’s network. Each player chooses a \emph{fast} or \emph{slow} route. The slow route yields a constant utility of $0.25$. The fast route yields utility $1.5 - 0.5 \cdot (\text{number of players choosing the fast route})$, reflecting congestion.

\subsection{Additional Experimental Results}

\Cref{fig:additional} reports additional experiments on the Chicken game and Pigou’s network. \textcolor[HTML]{6c3483}{MASE} consistently outperforms the baselines in both coalition exploitability and social welfare. In Pigou’s network, purely self-interested players overuse the fast route, which in equilibrium becomes slow. By contrast, when players form coalitions and consider average utility within a coalition, they share the routes so that everyone is better off.

\begin{figure}
    \centering
    \includegraphics[width=0.95\linewidth]{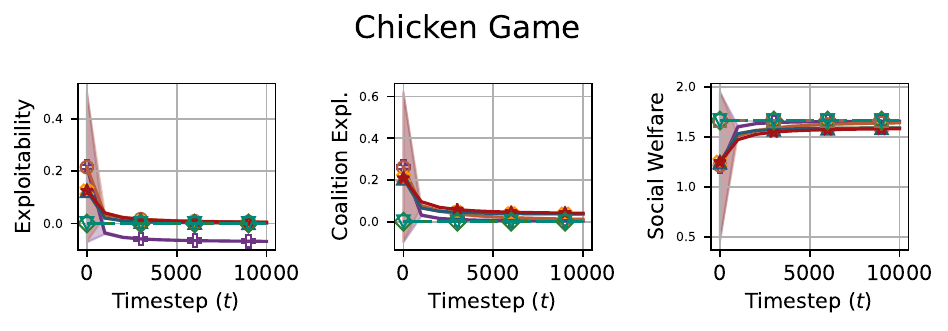}
    \includegraphics[width=0.95\linewidth]{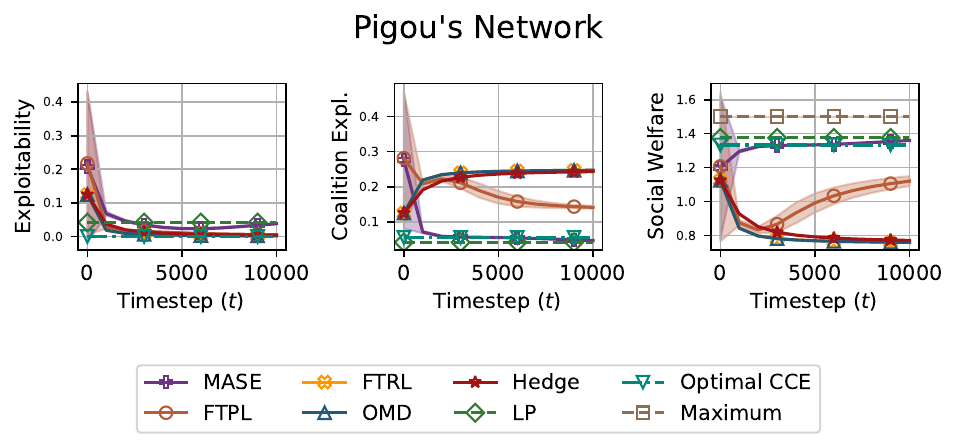}
    \caption{\textcolor[HTML]{2e7d32}{LP} refers to the linear program in \Cref{section:lp}. \textcolor[HTML]{8C6E54}{Maximum} marks the maximum social welfare. \textcolor[HTML]{6c3483}{MASE} outperforms the baselines in both games in terms of coalition exploitability and social welfare.}
    \label{fig:additional}
\end{figure}

\Cref{fig:tradeoff-2} shows the trade-off between exploitability and social welfare in the Chicken game and Pigou’s network.

\begin{figure}[h!]
    \centering
    \includegraphics[width=0.95\linewidth]{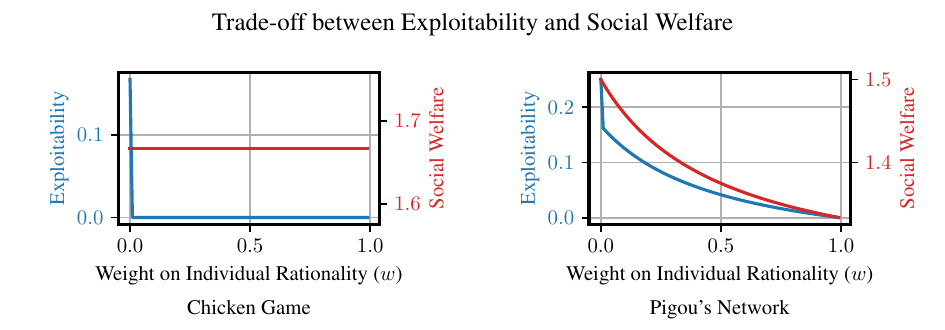}
    \caption{The trade-off between exploitability and social welfare in the Chicken game and the Pigou’s network.}
    \label{fig:tradeoff-2}
\end{figure}

\Cref{fig:time} reports the runtime of the algorithm for polymatrix games with varying numbers of players, action set sizes, and interaction densities.

\begin{figure}[h!]
    \centering
    \includegraphics[width=0.95\linewidth]{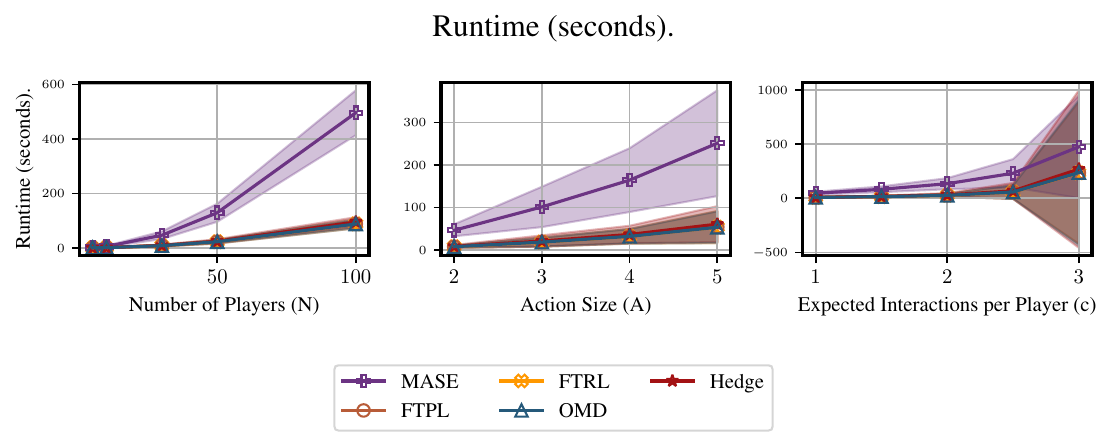}
    \caption{The runtime of the algorithm in different polymatrix games.}
    \label{fig:time}
\end{figure}

\section{Polymatrix Games}
\label{section:polymatrix}

In this section, we present experimental details for MASE on games with a larger number of players. We select polymatrix games as the benchmark for these large-scale experiments. This choice is motivated by their inherent graphical structure, which allows for the efficient generation of instances with a low treewidth of their \Cref{def:utility-dependence-graph}.

We begin with the formal definition. A polymatrix game has a corresponding undirected graph $\cG^U=(\cV^U, \cE^U)$, with $\cV^U=[N]$. For any joint action $\ba\in\cA$, the utility of any player $i$ is defined as:
\begin{align}
    \cU_i(\ba)\coloneqq \sum_{(i,j)\in \cE^U} \cU_{i,j}(a_i, a_j),
\end{align}
where $\cU_{i,j}\colon \cA_i\times\cA_j\to [0, 1]$ represents the interaction between players $i$ and $j$. In other words, only players who are connected in $\cG^U$ interact, and a player's total utility is the summation of these pairwise interactions.

If we construct the \Cref{def:utility-dependence-graph} directly, then the tree decomposition may explode unwillingly, \emph{e.g.}, \Cref{fig:polymatrix} (a). We can see that the treewidth of $\cG^U$ is one while the treewidth of the \Cref{def:utility-dependence-graph} is three.

Constructing the \Cref{def:utility-dependence-graph} directly from the polymatrix game can cause its treewidth to explode. For example, in \Cref{fig:polymatrix} (a), the original graph $\cG^U$ has a treewidth of one, while the resulting \Cref{def:utility-dependence-graph} has a treewidth of three.

To prevent this, we construct a strategically equivalent game (note that this new game is \emph{not} a polymatrix game). This construction explicitly models the pairwise interactions as new players:
\begin{itemize}
    \item For any original edge $(i,j)\in\cE^U$, we introduce two edge players, $e_{i,j}$ and $e_{j,i}$.
    \item Each edge player $e_{i,j}$ has a singleton action set, $\abr{\cA_{e_{i,j}}}=1$, meaning it has only a single strategy.\item The utility function of an edge player $e_{i,j}$ is defined as the original interaction utility: $\tilde\cU_{e_{i,j}}=\cU_{i,j}$.
    \item The utility function of an original vertex player $i$ (one of the original $N$ players) is now a constant zero: $\tilde\cU_i\equiv 0$.
\end{itemize}
This transformation is illustrated in \Cref{fig:polymatrix} (b). The \Cref{def:utility-dependence-graph} for this new game, shown on the right of \Cref{fig:polymatrix} (b), now has a treewidth of $\max\rbr{\tw{\cG^U}, 2}$. This method effectively bounds the treewidth and avoids the undesirable explosion.

\begin{figure}[t]
    \centering
    \includegraphics[width=0.5\linewidth]{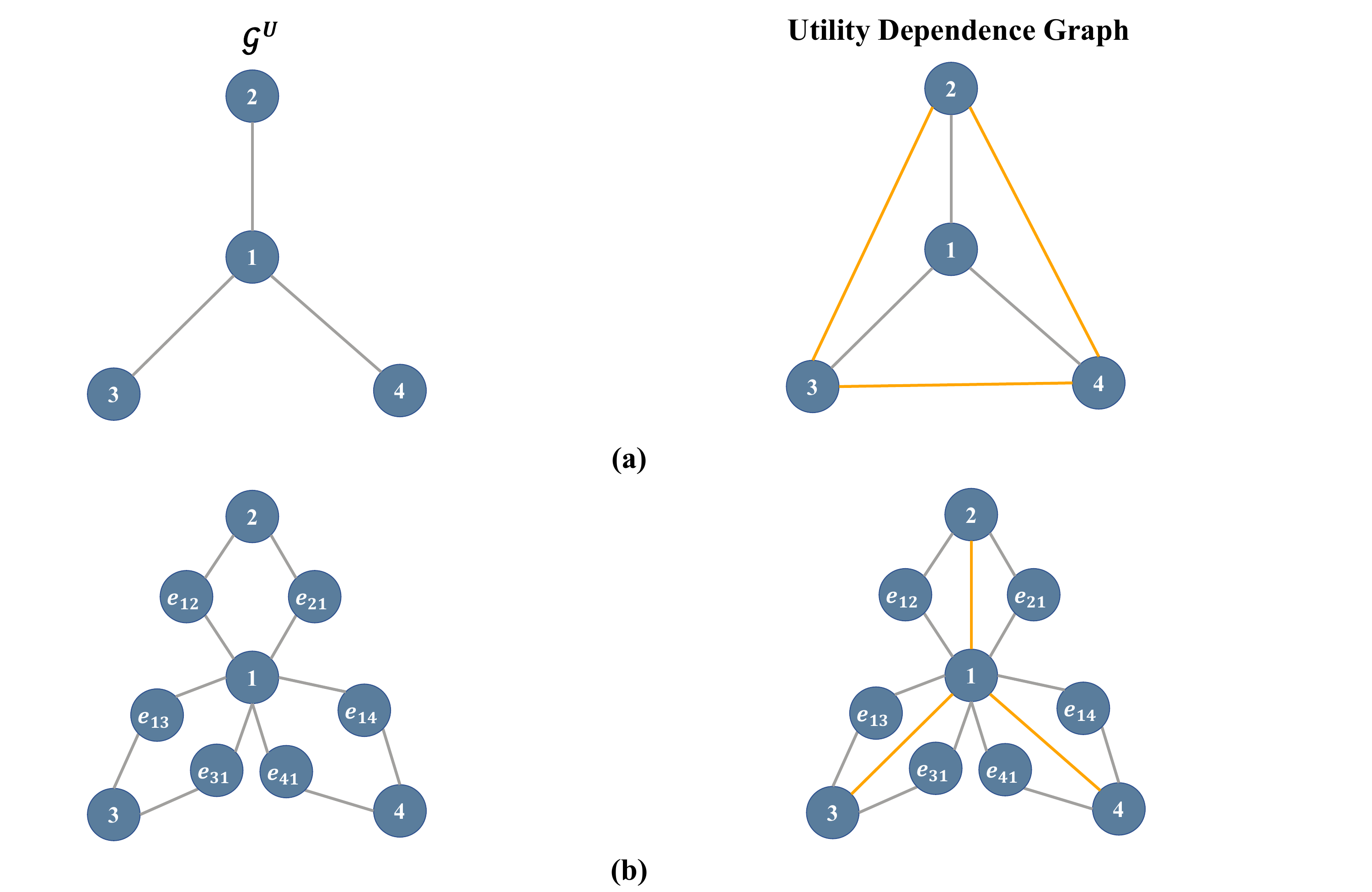}
    \caption{(a) The original graph $\cG^U$ corresponding to the polymatrix game (left) and the \Cref{def:utility-dependence-graph}. (b) The strategically equivalent game and its \Cref{def:utility-dependence-graph}.}
    \label{fig:polymatrix}
\end{figure}

Next, we show that the new game and the original polymatrix game are strategically equivalent. In other words, for any joint strategy $\pi\in\Delta^{\cA}$, the maximum average deviation gain, $\max_{S\in \cS}\allowbreak \max_{\hat\ba_S\in\cA_S}\allowbreak \frac{1}{|S|}\sum_{i\in S}\allowbreak \EE_{\ba\sim \pi}\sbr{\cU_i\rbr{\hat\ba_S, \ba_{-S}} - \cU_i\rbr{\ba} }$, does not change. Recall that since the edge players have only a single action, $\pi\in\Delta^{\cA}$ (a distribution over the original players' joint actions) is sufficient to specify the joint strategy in both games.

\begin{lemma}
\label{lemma:equivalent-game}
The new game described above is equivalent to the original polymatrix game, when $\tilde\cS=\cbr{S\cup \cbr{e_{i,j}}_{i\in S\wedge (i,j)\in \cE^U}}_{S\in\cS}$. Formally, for any joint strategy $\pi\in\Delta^{\cA}$, we have \begin{align*} 
    &\max_{S\in \cS} \max_{\hat\ba_S\in\cA_S} \frac{1}{|S|}\sum_{i\in S} \EE_{\ba\sim \pi}\sbr{\cU_i\rbr{\hat\ba_S, \ba_{-S}} - \cU_i\rbr{\ba} }\\
    =&\max_{\tilde S\in \tilde\cS} \max_{\hat\ba_{\tilde S}\in\cA_{\tilde S}} \frac{1}{|\tilde S\cap [N]|}\sum_{i\in \tilde S} \EE_{\ba\sim \pi}\sbr{\tilde \cU_i\rbr{\hat\ba_{\tilde S\cap [N]}, \ba_{-(\tilde S\cap [N])}} - \tilde\cU_i\rbr{\ba} }. 
\end{align*} 
\end{lemma}

The proof is postponed to the end of this section. This equivalence allows us to solve the new game instead of the original one. Our algorithm can minimize a more general objective, $\max_{S\in \cS}\allowbreak \max_{\hat\ba_S\in\cA_S}\allowbreak w_S\allowbreak\sum_{i\in S}\allowbreak \EE_{\ba\sim \pi}\sbr{\cU_i\rbr{\hat\ba_S, \ba_{-S}} - \cU_i\rbr{\ba} }$, for any weight vector $\bw\in\RR^S$.\footnote{Both the implementation and the proof only use the linearity of the objective. Hence, any weighted sum can fit into the framework.} We can therefore apply our algorithm to this new, strategically equivalent game.

\subsection{Experimental Details}

We generate random polymatrix games using the following procedure:
\begin{itemize}
    \item Each pair of players $(i,j)$ is connected independently with probability $\frac{c}{N-1}$, where $N$ is the total number of players. This results in an expected degree of $c$ for each player in $\cG^U$.
    \item For each connected pair $(i,j)$, the interaction utilities $\cU_{i,j}(a_i,a_j)$ are sampled independently and uniformly from $[0,1]$ for all action pairs $a_i\in\cA_i,a_j\in\cA_j$. These pairwise utilities are then normalized according to the formula:
\begin{align*}
    \frac{\cU_{i,j}(a_i,a_j) - \min_{k\in[N], \hat\ba\in\cA} \cU_k(\hat\ba)}{\max_{k\in[N], \hat\ba\in\cA} \cU_k(\hat\ba) - \min_{k\in[N], \hat\ba\in\cA} \cU_k(\hat\ba)}.
\end{align*}
This process ensures that the final total utility $\cU_i(\ba)$ for any player $i\in[N]$ and joint action $\ba\in\cA$ falls within the within the range $[0,1]$.
\end{itemize}

Consistent with the experiments on small games, we average the results over \emph{100 runs} for each hyper-parameter setting (using seeds 0–99). All algorithms use a learning rate of $\eta=0.01$, and error bars represent $1~\sigma$. For these larger games, we set the number of timesteps to $T=100,000$ and a uniform action set size $\abr{\cA_i}=A$ for all players $i\in [N]$.

The hyper-parameters for the ablation studies are as follows:
\begin{itemize}
\item \textbf{Ablation on $N$}: $A=2$ and $c=1$.
\item \textbf{Ablation on $A$}: $N=30$ and $c=1$.
\item \textbf{Ablation on $c$}: $N=30$ and $A=2$.
\end{itemize}

Furthermore, to accelerate the algorithm, without loss of generality, we only need to consider $\cS=\cbr{\cbr{i}}_{i\in [N]}\cup \cbr{\cbr{i,j}\given i,j\in [N], (i,j)\in\cG^U}$ to minimize the coalition exploitability for any coalitions with no more than two players. In other words, for coalitions of two players, we only need to consider the case when they are connected in $\cG^U$. As shown in the following lemma.
\begin{lemma}
\label{lemma:connected-players-coalition}
    For any joint strategy $\pi\in\Delta^{\cA}$, by letting $\cS=\cbr{\cbr{i}}_{i\in [N]}\cup \cbr{\cbr{i,j}\given i,j\in [N]\wedge (i,j)\in\cG^U}$, we have 
    \begin{align*} 
    &\max_{S\in \cS} \max_{\hat\ba_S\in\cA_S} \frac{1}{|S|}\sum_{i\in S} \EE_{\ba\sim \pi}\sbr{\cU_i\rbr{\hat\ba_S, \ba_{-S}} - \cU_i\rbr{\ba} }\\
    =&\max_{\substack{S\in \cbr{\cbr{i}}_{i\in [N]}\\\cup \cbr{\cbr{i,j}\given i,j\in [N]\wedge i\not=j}}} \max_{\hat\ba_S\in\cA_S} \frac{1}{|S|}\sum_{i\in S} \EE_{\ba\sim \pi}\sbr{\cU_i\rbr{\hat\ba_S, \ba_{-S}} - \cU_i\rbr{\ba} }.
\end{align*}
\end{lemma}
The proof is postponed to the end of this section.

\subsection{Proof of the Auxiliary Lemma}

\restate{lemma:equivalent-game}
\begin{proof}
    For any $S\in\cS$, let $\tilde S$ be its correspondence in $\tilde \cS$. Then,
    \begin{align*}
        &\max_{\hat\ba_{\tilde S}\in\cA_{\tilde S}} \frac{1}{|\tilde S\cap [N]|}\sum_{i\in \tilde S} \EE_{\ba\sim \pi}\sbr{\tilde \cU_i\rbr{\hat\ba_{\tilde S\cap [N]}, \ba_{-(\tilde S\cap [N])}} - \tilde\cU_i\rbr{\ba} }\\
        \overset{(i)}{=}&\max_{\hat\ba_S\in\cA_S} \frac{1}{|S|}\sum_{i\in \tilde S} \EE_{\ba\sim \pi}\sbr{\tilde \cU_i\rbr{\hat\ba_S, \ba_{-S}} - \tilde\cU_i\rbr{\ba} }\\
        \overset{(ii)}{=}& \max_{\hat\ba_S\in\cA_S} \frac{1}{|S|}\sum_{i\in S}~\sum_{j\colon (i,j)\in \cE^U} \EE_{\ba\sim \pi}\sbr{\tilde \cU_{e_{i,j}}\rbr{\hat\ba_S, \ba_{-S}} - \tilde\cU_{e_{i,j}}\rbr{\ba} }\\
        \overset{(iii)}{=}&\max_{\hat\ba_S\in\cA_S} \frac{1}{|S|}\sum_{i\in S} \EE_{\ba\sim \pi}\sbr{\cU_i\rbr{\hat\ba_S, \ba_{-S}} - \cU_i\rbr{\ba} }.
    \end{align*}
    $(i)$ uses the fact that $|\cA_{e_{i,j}}|=1$ and $\tilde S\cap [N]=S$. $(ii)$ is because $\cU_i\equiv 0$ for any $i\in [N]$. $(iii)$ is by the definition of $\tilde \cU_{e_{i,j}}$ and $\tilde S$.
\end{proof}

\restate{lemma:connected-players-coalition}
\begin{proof}
    For any disconnected players $i,j$ and $S=\cbr{i,j}$, we can see that
    \begin{align*}
        &\max_{\hat\ba_S\in\cA_S} \frac{1}{|S|}\sum_{k\in S} \EE_{\ba\sim \pi}\sbr{\cU_k\rbr{\hat\ba_S, \ba_{-S}} - \cU_k\rbr{\ba} }\\
        \leq& \max_{\hat\ba_S\in\cA_S} \max_{k\in S} \EE_{\ba\sim \pi}\sbr{\cU_k\rbr{\hat\ba_S, \ba_{-S}} - \cU_k\rbr{\ba} }\\
        =&\max_{k\in S} \max_{\hat\ba_S\in\cA_S} \EE_{\ba\sim \pi}\sbr{\cU_k\rbr{\hat\ba_S, \ba_{-S}} - \cU_k\rbr{\ba} }\\
        =& \max_{k\in S} \max_{\hat\ba_S\in\cA_S} \EE_{\ba\sim \pi}\sbr{\sum_{k'\in [N]\colon (k,k')\in\cG^U} \cU_{k, k'}\rbr{\hat\ba_S, \ba_{-S}} - \cU_{k, k'}\rbr{\ba} }\\
        \overset{(i)}{=}& \max_{k\in S} \max_{\hat\ba_S\in\cA_S} \EE_{\ba\sim \pi}\sbr{\sum_{k'\in [N]\colon (k,k')\in\cG^U} \cU_{k, k'}\rbr{\hat\ba_k, \ba_{-k}} - \cU_{k, k'}\rbr{\ba} }\\
        =&\max_{k\in S} \max_{\hat\ba_S\in\cA_S} \EE_{\ba\sim \pi}\sbr{\cU_k\rbr{\hat\ba_k, \ba_{-k}} - \cU_k\rbr{\ba} }.
    \end{align*}
    $(i)$ is because $k'\not\in S$ since $k\in S=\cbr{i,j}$ and $i,j$ are not connected. Therefore, since the coalition exploitability of $S$ is upper bounded by the maximum of that of coalitions $\cbr{i}$ and $\cbr{j}$, we do not need to consider $\cbr{i,j}$.

    Actually, the argument can be generalized to coalitions of any size $M$. If we want to consider the coalition exploitability for coalitions no more than size $M$, then we only need to consider all connected coalitions of size no more than size $M$ by an induction similar to the proof above.
\end{proof}

\section{Omitted Proofs for \Cref{section:tradeoff}}
\label{section:proof-tradeoff}

\restate{lemma:pareto-frontier-concavity}

\begin{proof}
    First, note that the set of $\epsilon$-CCEs is a polytope. In particular, it can be written as
\[
\epsilon\text{-CCE}=\cbr{\one^\top \pi = 1 \ \wedge\ \bB\pi\leq \epsilon \ \wedge\ \pi\geq 0},
\]
where $\one\in \RR^{\cA}$ denotes the all-ones vector. Moreover, social welfare is linear in $\pi$, and can be expressed as $\bc^\top \pi$ for some vector $\bc\in\RR^{\cA}$. Hence, for a fixed $\epsilon$, computing an optimal $\pi\in\epsilon\text{-CCE}$ is equivalent to the linear program
    \begin{align}
        &\max \bc^\top\pi\notag\\
        \text{s.t.}\quad& \one^\top \pi = 1\tag{Primal}\\
        &\bB\pi\leq \epsilon\notag\\
        &\pi\geq 0.\notag
    \end{align}
    As the feasible set expands with increasing $\epsilon$, the EWF is non-decreasing. The dual of the linear program above is
    \begin{align}
        \min_{\lambda,\bmu}\quad & \lambda + \epsilon\cdot \one^\top\bmu\notag\\
        \text{s.t.}\quad & \lambda\,\one + \bB^\top \bmu \geq \bc,\tag{Dual}\\
        & \bmu \geq 0,\qquad \lambda \in \mathbb{R}\ \text{(free)}.\notag
    \end{align}
    The primal is feasible and bounded for every $\epsilon\geq 0$ (every Nash equilibrium exists \citep{nash1950equilibrium-nash-def} and is a CCE), so by strong duality the primal and dual attain the same optimal value. Importantly, the dual feasible region does not depend on $\epsilon$, and we denote it by $\cD$. Since the dual optimum is finite by strong duality, it is attained at a vertex of $\cD$. Because $\cD$'s vertex set $\text{Vertex}(\cD)$ is finite, we obtain
    \begin{align*}
        \ewf^{\rm CCE}(\epsilon)=\min_{(\lambda,\bmu)\in \text{Vertex}(\cD)} \lambda + \epsilon\cdot \one^\top\bmu.
    \end{align*}
    Finally, $\ewf^{\rm CCE}$ is the pointwise minimum of finitely many linear functions of $\epsilon$, and is therefore piecewise linear and concave. \qedhere
\end{proof}

\restate{lemma:slope-unbounded}

\begin{proof}
    In the Stag Hunt, the slope can be zero because a social-welfare-maximizing joint strategy is itself a Nash equilibrium. Next, we give an example in which the slope is unbounded.

    \begin{table}[h]
    \centering
    \begin{tabular}{|c|c|c|}
    \hline
     & \textbf{Confess (C)} & \textbf{Defect (D)} \\
    \hline
    \textbf{Confess (C)} & $(0.6,0.6)$ & $(0.2-\zeta,1)$ \\
    \hline
    \textbf{Defect (D)} & $(1,0.2-\zeta)$ & $(0.2,0.2)$ \\
    \hline
    \end{tabular}
    \caption{Utility matrix for a variant of Prisoner’s Dilemma, where $\zeta>0$ is a constant. Each entry $(a,b)$ gives the payoff to the row player ($a$) and the column player ($b$).}
    \label{table:PD-variant}
\end{table}

    Consider the variant of Prisoner’s Dilemma in \Cref{table:PD-variant}. It is straightforward to verify that the unique CCE remains $(D,D)$. Now consider the following strategy\footnote{In fact, this strategy is not correlated, so \Cref{lemma:slope-unbounded} extends to $\ewf^{\rm NE}$ and $\ewf^{\rm CE}$ as well.}, which is a linear mixture of the pure strategies $(D,D)$ and $(C,D)$.
    \begin{align*}
        (1-p)\cdot(D,D) + p\cdot (C, D).
    \end{align*}
    Its social welfare equals
    \begin{align*}
        0.4(1-p) + (1.2-\zeta)p=0.4 + (0.8-\zeta)p,
    \end{align*}
    while its exploitability is $\zeta\cdot p$. Hence, the slope at $\epsilon=0$ is at least
    \begin{align*}
        \frac{0.4 + (0.8-\zeta)p - 0.4}{\zeta\cdot p - 0.0} = \frac{0.8-\zeta}{\zeta}.
    \end{align*}
    Letting $\zeta\to 0$ makes this ratio diverge, so the slope can be arbitrarily large. \qedhere
\end{proof}

\restate{lemma:equivalence-tradeoff}

\begin{proof}
    For any $\epsilon > 0$, let $\pi^*$ be the strategy that maximizes social welfare subject to its exploitability being at most $\epsilon$. Let $g^* = \max_{\hat\ba\in\cA} \frac{1}{N}\sum_{i=1}^N \EE_{\ba\sim \pi^*}\sbr{\cU_i\rbr{\hat\ba} - \cU_i\rbr{\ba}}$ be the maximum gain for the grand coalition $[N]$. Let $w=\frac{g^*}{\epsilon + g^*}$. Then, by construction, the objective value for $\pi^*$ under \Cref{eq:weighted-mase} is:
    \begin{align*}
        \max_{S\in \cS} \max_{\hat\ba_S\in\cA_S} \frac{w_S}{|S|}\sum_{i\in S} \EE_{\ba\sim \pi^*}[\dots]
        &= \max\left( w \cdot (\text{exploitability}), (1-w) \cdot g^* \right) \\
        &\leq \max\left( w \epsilon, (1-w) g^* \right) = \frac{\epsilon g^*}{\epsilon + g^*}.
    \end{align*}
    Any strategy $\hat{\pi}$ with exploitability $>\epsilon$ would have an objective value $> w\epsilon = \frac{\epsilon g^*}{\epsilon + g^*}$, which is worse than the value $\pi^*$ achieves. Therefore, any optimal solution to \Cref{eq:weighted-mase} must have exploitability at most $\epsilon$. Since $\pi^*$ by definition maximizes social welfare (i.e., minimizes the coalition gain $g^*$) among all strategies in this set, it must also be an optimal solution to \Cref{eq:weighted-mase}.

    Conversely, for any $w\in [0, 1)$, let $\pi^*$ be the corresponding strategy that optimizes \Cref{eq:weighted-mase}. Let its exploitability be
    \begin{align*}
        \epsilon=\max_{i\in [N]} \max_{\hat a_i\in\cA_i} \EE_{\ba\sim \pi^*}\sbr{\cU_i\rbr{\hat a_i, \ba_{-i}} - \cU_i\rbr{\ba}}.
    \end{align*}
    We will show by contradiction that no strategy $\pi'$ exists such that $\text{exploitability}(\pi') \leq \epsilon$ and $SW(\pi') > SW(\pi^*)$ (which implies $g' < g^*$, where $g'$ is the gain for the grand coalition under $\pi'$).
    
    Suppose such a $\pi'$ exists. We analyze two cases:

    \paragraph{Case 1: $\text{exploitability}(\pi') < \epsilon$.}
    Since $\pi'$ has both strictly lower exploitability than $\pi^*$ and $g' < g^*$ (higher social welfare), its objective value is $\max(w \cdot \text{exploitability}(\pi'), (1-w)g')$. This is strictly less than $\max(w\epsilon, (1-w)g^*)$, which is the objective value of $\pi^*$. This contradicts the optimality of $\pi^*$.

    \paragraph{Case 2: $\text{exploitability}(\pi') = \epsilon$.}
    If $\epsilon > 0$, choose a small $\delta>0$ and consider the mixed strategy $\pi_{new} = (1-\delta)\pi'+\delta \pi''$, where $\pi''$ is an arbitrary CCE, which is guaranteed to exist \citep{nash1950equilibrium-nash-def}. For any $i\in [N]$ and $\hat a_i\in\cA_i$, we have:
    \begin{align*}
        &\EE_{\ba\sim \pi_{new}}\sbr{\cU_i\rbr{\hat a_i, \ba_{-i}} - \cU_i\rbr{\ba}}\\
        =&(1-\delta)\EE_{\ba\sim \pi'}\sbr{\cU_i\rbr{\hat a_i, \ba_{-i}} - \cU_i\rbr{\ba}} + \delta \EE_{\ba\sim \pi''}\sbr{\cU_i\rbr{\hat a_i, \ba_{-i}} - \cU_i\rbr{\ba}}\\
        \overset{(i)}{\leq}& (1-\delta)\EE_{\ba\sim \pi'}\sbr{\cU_i\rbr{\hat a_i, \ba_{-i}} - \cU_i\rbr{\ba}} \\
        \leq& (1-\delta)\epsilon.
    \end{align*}
    Step $(i)$ holds because $\pi''$ is a CCE, so its exploitability $\EE_{\pi''}[\dots]$ is $\leq 0$.
    Since $\epsilon > 0$, the new strategy $\pi_{new}$ has $\text{exploitability}(\pi_{new}) < \epsilon$. By continuity, for sufficiently small $\delta$, $SW(\pi_{new})$ remains strictly higher than $SW(\pi^*)$ (since $SW(\pi') > SW(\pi^*)$). This puts us in Case 1, which leads to a contradiction.

    If $\epsilon \leq 0$, then $\text{exploitability}(\pi^*) \leq 0$. The objective value for $\pi^*$ is $\max(w\epsilon, (1-w)g^*) = (1-w)g^*$ (since $w\epsilon \leq 0$ and $(1-w)g^* \geq 0$ by definition). The hypothetical strategy $\pi'$ has $\text{exploitability}(\pi') = \epsilon \leq 0$ and $g' < g^*$. Its objective value is $\max(w\epsilon, (1-w)g') = (1-w)g'$. Since $g' < g^*$ and $w<1$, the objective value of $\pi'$ is strictly less than that of $\pi^*$, which contradicts the optimality of $\pi^*$. %

    In all cases, the existence of such a $\pi'$ leads to a contradiction. Thus, $\pi^*$ must be a solution that maximizes social welfare for a given exploitability $\epsilon$.
\end{proof}

\restate{lemma:ewf-hardness}

\begin{proof}
    The proof follows from Theorem 5.10 of \citet{papadimitriou2008computing-against-hope}, which shows that for each of the succinct game classes listed above, deciding whether a \emph{dominant outcome} exists, \emph{i.e.}, whether there is a joint strategy at which every player attains that player's maximum utility, is {\sf NP}-hard. We reduce this decision problem to computing $\ewf^{\rm CCE}(0)$.

    Let $M \coloneqq \sum_{i=1}^N \max_{\ba\in\cA}\cU_i(\ba)$. The quantity $M$ can be computed efficiently in all of the above games \citep{papadimitriou2008computing-against-hope}. If a dominant outcome exists, then the corresponding joint strategy is a CCE, implying $\ewf^{\rm CCE}(0)=M$. Conversely, if no dominant outcome exists, then no CCE can achieve every player’s maximum utility simultaneously, and thus
    \begin{align*}
        \ewf^{\rm CCE}(0) < M.
    \end{align*}
    Therefore, computing $\ewf^{\rm CCE}(0)$ decides whether a dominant outcome exists, establishing {\sf NP}-hardness. \qedhere
\end{proof}

\end{document}